\documentclass[11pt,a4paper]{article}
\newcommand{\dowod}{\begin{proof}}
\newcommand{\koniec}{\end{proof}}

\usepackage{amssymb, amsmath, amsthm}
\usepackage{xspace}
\usepackage{calc}
\usepackage[sort&compress,numbers]{natbib}
\usepackage{enumitem}
\usepackage{times}
\usepackage[T1]{fontenc}
\usepackage{hyperref}
\hypersetup{
    colorlinks=true,       
    linkcolor=blue,        
    citecolor=red,         
    filecolor=magenta,     
    urlcolor=cyan,         
    linktocpage=true
}

\usepackage{amssymb}        
\usepackage{latexsym}
\usepackage{algorithm}
\usepackage[noend]{algpseudocode}
\usepackage{graphicx}
\usepackage{verbatim}

\usepackage{tikz}
\usepackage{subcaption}
\usetikzlibrary{calc}
\usetikzlibrary{decorations.markings,decorations.pathmorphing,positioning,intersections}

\newtheorem{theorem}{Theorem}
\newtheorem{lemma}{Lemma}
\newtheorem{fact}{Fact}
\newtheorem{observation}{Observation}
\newtheorem{definition}{Definition}

\newtheorem{claim}{Claim}
\newtheorem{corollary}{Corollary}
\newcommand{\np}{$\mathcal{NP}$}

\usepackage{hyperref}

\title{Triangle-free $2$-matchings \thanks{Partially supported by Polish National Science Center grant 2018/29/B/ST6/02633.}}

\author{Katarzyna Paluch\\
Institute of Computer Science,  University of Wroc{\l}aw \\
{\tt abraka@cs.uni.wroc.pl}}

\begin{document}
\maketitle

\begin{abstract}

We consider the problem of finding a maximum size triangle-free $2$-matching in a graph $G=(V,E)$. A (simple) $2$-matching is any subset of the edges such that each vertex is incident to at most two  edges from the subset. The first polynomial time algorithm for this problem was given by Hartvigsen in 1984 in his PhD thesis and its improved version has been recently published in a journal.   We present a different, significantly simpler algorithm with a relatively short proof of correctness. Our algorithm with running time $O(|V||E|)$  is additionally faster than the one by Hartvigsen having running time  $O(|V|^3|E|^2)$.

It has been proven before that for any triangle-free $2$-matching $M$ which is not maximum the graph contains an $M$-augmenting path, whose application to $M$ results in a bigger triangle-free $2$-matching.   A new observation is that the search for an augmenting path $P$ can be restricted to so-called {\em amenable} paths that go through any triangle $t$ contained in $P\cup M$ a limited number of times. Amenable paths can be  characterised with the aid of {\em half-edges}. A {\em half-edge} of edge $e$ is, informally speaking, a half  of $e$ containing exactly one of its endpoints. Each half-edge serves also as a {\em hinge} - a connector between one pair of edges on an alternating path. To find an amenable augmenting path we thus dynamically
remove and re-add  half-edges to forbid or allow some edges to be followed by certain others. This  can also be interpreted as using gadgets, which are not  fixed during any augmentation phase, but are  changing according to the currently discovered state of reachability by amenable paths. 

The existence of amenable augmenting paths follows from our decomposition theorem for  triangle-free $2$-matchings.
This decomposition theorem is  largely the same as the decomposition from  versions 1-6 of this paper and  is moreover simpler and stronger than the one given by Kobayashi and 
Noguchi in \cite{NoguchiSosa}.

\end{abstract}

\section{Introduction}

A subset~$M$ of edges of an undirected simple graph $G=(V,E)$ is a (simple) \emph{2-matching} if every vertex is incident to at most two  edges of~$M$. For brevity, a simple $2$-matching will further be called a $2$-matching. A $2$-matching of maximum size can be computed in polynomial time by a reduction to the classical matching problem. A $2$-matching is called \emph{triangle-free} if it does not contain any cycle of length three. The first polynomial time algorithm for computing a maximum size triangle-free $2$-matching was given by Hartvigsen \cite{Hartvigsen1984} in his PhD thesis in 1984. This algorithm and its analysis are very complicated. 
A journal version with the modified algorithm from the thesis appeared recently in \cite{Hartvigsen2024}. It is still very long ($82$ pages) and complicated and the algorithm given there has $O(|V|^3|E|^2)$ running time.  We present an alternative and significantly simpler algorithm with a relatively short proof of correctness and $O(|V||E|)$ running time.

 More generally,  a \emph{$C_k$-free} $2$-matching is one without any cycle of length at most ~$k$. We refer to cycles of length three as \emph{triangles}.
The goal of the \emph{$C_k$-free $2$-matching problem} is to find a $C_k$-free $2$-matching of maximum size. Observe that the $C_k$-free $2$-matching problem for $n/2 \leq k < n$, where $n$ is the number of vertices in the graph, is equivalent to finding a Hamiltonian cycle, and thus \np-hard.
The case of  $k=3$ is also called the \emph{triangle-free $2$-matching problem}. For $k\geq 5$ Papadimitriou \cite{CornuejolsPulleyblank1980} showed that the problem is \np-hard. The complexity of the $C_4$-free $2$-matching problem is unknown.

In the weighted version of the problem,  each edge $e$ is associated with a  nonnegative weight $w(e)$ and we are interested in finding a $C_k$-free $2$-matching of maximum  weight, where the weight of a $2$-matching $M$ is defined as the sum of weights of edges belonging to $M$.  
Vornberger~\cite{Vornberger1980} showed that the weighted $C_4$-free $2$-matching problem is \np-hard. 

{\em Motivation. } Matching problems are among the most well-studied and key problems of combinatorial optimization. The triangle-free $2$-matching problem  is a classical problem of matching theory, better understanding of which has bearing on other problems and may enable further extensions. 
Additionally,  it has applications in heavily researched problems in theoretical computer science, namely,  traveling salesman problems (\cite{BlaserRam2005}, \cite{BermanKarpinski2006}, \cite{AdamaszekMP18}) and  problems related to finding a smallest $2$-edge-connected spanning subgraph \cite{Kobayashi2023},  \cite{Calvostoc}.
In fact, the computation of a maximum triangle-free $2$-matching (needed for a minimum triangle-free $2$-edge cover) is a crucial part of approximation algorithms for the $2$-edge-connected spanning problem in \cite{Kobayashi2023} and \cite{Calvostoc}. It is somewhat strange but this information is completely omitted in the conference version of \cite{Calvostoc} and in the ArXiv version mentioned only in the appendix.
More generally, the $C_k$-free $2$-matching problems can also be used for  increasing the vertex-connectivity (see \cite{FisherEtAl1979,BercziKobayashi2012,BercziVegh2010} for more details). A good survey of these applications has been given by Takazawa~\cite{Takazawa2017a}.

{\em Our results. } When computing a maximum triangle-free $2$-matching one can successively and exhaustively augment any triangle-free $2$-matching $M$. Russell \cite{Russell2001} proved that 
for a non-maximum triangle-free $2$-matching $M$ there exists in the graph  an augmenting path $P$ such that after its application to $M$ we obtain a new larger 
$2$-matching, which is triangle-free.  However, it was not known how to efficiently find such a path. We show that the search for an augmenting path that preserves the property of being triangle-free can be restricted to so called {\em amenable} paths.
It turns out that assuming that $M$ is a non-maximum triangle-free $2$-matching, there always exists an augmenting path $P$ that passes through any triangle $t \in P \cup M$ at most $|t \setminus M|$ times, i.e., if $t$ contains two edges in $M$, then $P$ can go through $t$ only once - an edge $e \in t \setminus M$ is immediately followed or preceded on $P$ by an edge in $t \cap M$. Amenable paths can be  characterised with the aid of {\em half-edges}.

A {\em half-edge} of edge $e$ is, informally speaking, a half  of $e$ containing exactly one of its endpoints. 
In the current context a half-edge may be also viewed as a {\em hinge} - a connecter between one pair of edges on an alternating path. Half-edges  have also been used in the Tutte reduction (\cite{Schrijver2003} p.526) of simple $2$-matchings to matchings (and also of simple $b$-matchings, which are further generalizations of matchings, to matchings). The main difference between their application in Tutte's reductions and ours is that in Tutte's reductions either both half-edges of a given edge are used or none, whereas we consider more  unconventional configurations and in particular, we admit using only one half-edge of an edge.
Half-edges of the new type have already been  introduced in \cite{PaluchEtAl2012} and used in several subsequent papers. Their application in the present paper to disconnecting pairs of edges is novel.

To find an amenable augmenting path we dynamically
remove and re-add  half-edges to forbid or allow some edges to be followed by certain others. 
In algorithms for maximum $C_k$-free $2$-matchings and related problems one usually uses gadgets or shrinking to enforce the computed matchings to have some properties. The gadgets are added 
to the graph and are fixed either throughout the algorithm or at least during the phase, in which an augmenting path is searched for. In our approach gadgets are dynamically changing according to the currently discovered state of reachability by amenable paths. 
The existence of amenable augmenting paths follows from our decomposition theorem for  triangle-free $2$-matchings.
This decomposition theorem is  in essence the same as the decomposition from  the  arxiv versions of this paper and  is moreover simpler and stronger than the one  in \cite{NoguchiSosa}.

 A good knowledge of the Edmonds' algorithm for maximum matchings    \cite{Edmonds,  LovaszPlummer2009} is a prerequisite for understanding our proofs.

{\em Related work. } Recently, Bosch-Calvo, Grandoni and Jabal Ameli presented a PTAS for the $C_3$-free $2$-matching problem in \cite{Bosch}. A simpler proof of correctness of this PTAS is given in \cite{NoguchiSosa}.  Kobayashi~\cite{Kobayashi2020} gave a polynomial algorithm for finding a maximum weight $2$-matching that does not contain any triangle from a given set of forbidden edge-disjoint triangles. For the weighted triangle-free $2$-matching problem in subcubic graphs, polynomial time algorithms were given by:
 Hartvigsen and Li~\cite{HartvigsenLi2012},  Kobayashi~\cite{Kobayashi2010} and Paluch and Wasylkiewicz \cite{PaluchWasylkiewicz2020}. (A graph is called {\em cubic} if its every vertex has degree $3$ and is   called {\em subcubic} if its every vertex has degree at most $3$. ) One can also consider non-simple $2$-matchings, in which every edge $e$ may occur in more than one copy. 
Efficient algorithms for triangle-free non-simple $2$-matchings (such $2$-matchings may contain $2$-cycles) were devised by Cornu\'ejols and Pulleyback \cite{CornuejolsPulleyblank1980, CornuejolsPulleyblank1980a}, Babenko, Gusakov and Razenshteyn \cite{BabenkoEtAl2010}, and Artamonov and Babenko \cite{ArtamonovBabenko2018}.

 Polynomial time algorithms for  the $C_4$-free $2$-matching problem in bipartite graphs were shown by Hartvigsen~\cite{Hartvigsen2006}, Pap~\cite{Pap2007}, Paluch and Wasylkiewicz \cite{PaluchWasylESA} and analyzed by Kir\'aly \cite{Kiraly1999}. As for the weighted version of the square-free $2$-matching problem in bipartite graphs it was proven to be \np-hard \cite{Geelen1999, Kiraly2009} and solved by  Makai~\cite{Makai2007}, Takazawa~\cite{Takazawa2009} and Paluch and Wasylkiewicz \cite{PaluchWasylESA} for the case  when  the weights of edges are vertex-induced on every square of the graph. When it comes to the square-free $2$-matching problem in general graphs, 
Nam~\cite{Nam1994} constructed a complex algorithm for it for graphs, in which all squares are vertex-disjoint.
 B\'erczi and Kobayashi \cite{BercziKobayashi2012} showed that the weighted square-free $2$-matching problem 
is \np-hard for general weights even if the given graph is cubic, bipartite and planar.

Other results for related matching problems appeared  in \cite{Pap2009,Takazawa2017,Takazawa2017a}.


\section{Preliminaries}
Let $G=(V,E)$ be an undirected graph with vertex set~$V$ and edge set~$E$. We denote the number of vertices of $G$ by $n$ and the number of edges of $G$ by $m$. We denote a vertex set of $G$ by $V(G)$ and an edge set by $E(G)$. We assume that all graphs are {\bf \em simple}, i.e., they contain neither loops nor parallel edges. We denote an edge connecting vertices $v$ and $u$ by $(v,u)$. A {\bf \em path} of graph~$G$ is a sequence $P=(v_0, \ldots, v_l)$ for some $l\geq 1$ such that $(v_i,v_{i+1})\in E$ for every $i\in\{0,1,\ldots,l-1\}$. We refer to $l$ as the {\bf \em length} of $P$. A {\bf \em cycle} of graph~$G$ is a sequence  $c=(v_0, \ldots, v_{l-1})$ for some $l \geq 3$ of pairwise distinct vertices of~$G$ such that $(v_i,v_{(i+1) \bmod l})\in E$ for every $i\in\{0,1,\ldots,l-1\}$. We refer to $l$ as the {\bf \em length} of $c$. We will sometimes treat a path or a cycle as an edge set and sometimes as a sequence of edges.
For an edge set $F\subseteq E$ and $v\in V$, we denote by $\deg_F(v)$ the number of edges of~$F$ incident to $v$. For any two edge sets $F_1,F_2\subseteq E$, the symmetric difference $F_1\oplus F_2$ denotes $(F_1\setminus F_2)\cup(F_2\setminus F_1)$.

For a natural number $t$, we say that an edge set $F\subseteq E$ is a {\bf \em $t$-matching} if $\deg_{F}(v)\leq t$ for every $v\in V$. 
$t$-matchings belong to a wider class of $b$-matchings, where for every vertex $v$ of $G$, we are given a natural number $b(v)$ and a subset of edges is a {\bf \em $b$-matching} if every vertex $v$ is incident to at most $b(v)$ of its edges. A $b$-matching of $G$ of maximum weight can be computed in polynomial time. We refer to Lov\'asz and Plummer~\cite{LovaszPlummer2009} for further background on $b$-matchings.

Let $M$ be a $b$-matching. We say that an edge $e$ is {\bf \em matched} (in $M$) if $e \in M$ and {\bf \em unmatched} (in $M$) otherwise. Additionally, an edge belonging to $M$ will be referred to as an {\bf \em $M$-edge} and an edge not belonging to $M$ as a {\bf \em non-$M$-edge}. We call a vertex $v$ {\bf \em deficient} or {\bf \em unsaturated (in $M$)} if $\deg_M(v)<b(v)$ and {\bf \em saturated} (in $M$) if $\deg_M(v)=b(v)$. {\bf \em An $M$-alternating path} $P$ is any sequence of vertices $(v_1,v_2,\ldots,v_k)$ such that edges on $P$ are alternately $M$-edges and non-$M$-edges and no edge occurs on $P$ more than once.
{\bf \em An $M$-alternating cycle} $C$ has the same definition as an $M$-alternating path except that $v_1=v_k$ and additionally $(v_{k-1},v_k)\in M$ iff $(v_1,v_2)\notin M$. Note that an $M$-alternating path or cycle may go through some vertices more than once but via different edges.
An $M$-alternating path is called {\bf \em $M$-augmenting} if it begins and ends with a non-$M$-edge and if it begins and ends with a deficient vertex. In the case when such an $M$-alternating path begins and ends with the same deficient vertex $v$, it has to hold that $deg_M(v) \leq b(v)-2$. We say that $M$ is a {\bf \em maximum} $b$-matching if there is no $b$-matching of $G$ with more edges than $M$. 
An {\bf \em application} of an $M$-alternating path or cycle $P$ to $M$ is an operation whose result is $M\oplus P$. 
An edge in a minor of $G$ is denoted by its preimage.

For a vertex $u$ matched in a $1$-matching $M$, by $M(u)$ we denote the vertex $v$ such that $(u,v) \in M$.
In all figures in the paper thick edges denote matched  edges, belonging to $M$ or $M'$, and thin ones unmatched ones.
All paths and almost all cycles occurring in the paper after Section \ref{decomp} are $M$-alternating or $M'$-alternating. 

An instance of the triangle-free $2$-matching problem consists of an undirected graph~$G=(V,E)$ and the goal is to find a maximum triangle-free $2$-matching of~$G$.

\section{Decomposition} \label{decomp}
Let $M$ denote a triangle-free $2$-matching. An $M$-alternating path or cycle $P$ is {\bf \em feasible} (w.r.t. $M$) if $M \oplus P$ is triangle-free and it is {\bf \em amenable} (w.r.t. $M$)
if for every triangle $t$ contained in $M \cup P$, it holds that some two edges of $t$ are consecutive on $P$. Note that an $M$-alternating path or cycle $P$ that is amenable is also feasible.
Consider the following example from Figure \ref{amenableexample}. Here the path $P'=(s,a,c,b,a,d,f,z)$ is  feasible and so is each of its even-length subpaths starting at $s$. On $P'$ we first go through an edge $e=(a,c)\in M$ of a triangle $t=(a,c,d)$ and only later and not immediately do we go through another edge $e'=(a,d) \notin M$ of $t$. Therefore, $P'$ is not amenable. However, we may reorder edges of $P'$ so that $P'$ becomes amenable. 

\begin{figure}[h]
\centering{\includegraphics[scale=0.9]{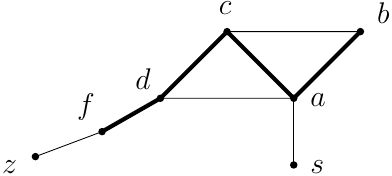}}
\caption{ An amenable path $(s,a,b,c,a,d,f,z)$.}
 \label{amenableexample}
\end{figure}

Let  $N$ denote a triangle-free  $2$-matching different from $M$. Suppose that we would like to decompose $M \oplus N$ into $M$-alternating paths and $M$-alternating cycles so that each path and cycle in the decomposition is amenable with respect to both $M$ and $N$. Let us call such a decomposition {\bf \em ideal}.
Let $e_1 \in M\setminus N, e_2 \in N\setminus M$ be two edges of $M \oplus N$ such that a vertex $v$ is their common endpoint. We say that they are {\bf \em bound (via $v$)} in the decomposition $D$ if they occur consecutively on some path or cycle of $D$. Clearly, any decomposition of $M \oplus N$ into alternating paths and cycles is uniquely determined by establishing for each edge $e=(v,w) \in M \oplus N$ with which edge it is bound via $v$ and with which via $w$, where  $e$  may also  not be bound via one or both of its endpoints. To {\bf \em bind} two edges in a decomposition $D$ means to place them on an alternating path or on an alternating cycle so that they are bound in $D$.

\begin{observation}\label{decompo1}
 Suppose that an alternating path or cycle $P$ contains an edge $e$ of a triangle $t$ in $M \oplus N$, which is neither preceded nor followed by another edge of $t$.
Let $e_1, e_2$ denote the other two of the three edges of $t$.
Then $P$ may belong  to some ideal decomposition $D$  only if

\begin{enumerate}
\item $|\{e_1, e_2\}\cap M|=1=|\{e_1, e_2\} \cap N|$,
\item $e_1, e_2$ are bound in $D$.
\end{enumerate}

\end{observation}

More generally, we can conclude:
\begin{corollary} \label{colideal}
Let $M$ and $N$ denote two triangle-free $2$-matchings.
Then, a decomposition $D$ of $M \oplus N$ is ideal if and only if for any triangle $t$ contained in $M \cup N$ it holds that (at least) two of its edges are bound in $D$.
\end{corollary}

Because of point $2$ of Observation \ref{decompo1} we need to pay attention to the following type of subgraphs.
Let $S$ denote a subgraph of $M \oplus N$ on vertices $a,b,c,d$ containing five edges: $(a,b), (a,d), (d,c) \in M$ and $(d,b), (b,c) \in N$.
We call such a subgraph {\bf \em dangerous}. We notice that if the edges $(a,b), (b,c)$ follow each other on an alternating path $P$ and $(b,c)$ is not followed by $(c,d)$, then  $P$ cannot belong to an ideal decomposition, because
the edge $(b,d)$ would have to be bound both with $(d,a)$ and with $(d,c)$, which is clearly impossible.  We call  $(a,b,c)$ a {\bf \em bumpy} path (of $S$.) 

\begin{figure}[h]
\centering{\includegraphics[scale=0.9]{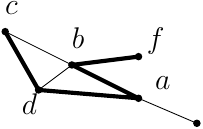}}
\caption{ A bumpy path $(a,b,c)$.}
 \label{bumpy}
\end{figure}

Now we can notice that avoiding or traversing bumpy paths in such a way that point $2$ of Observation \ref{decompo1} is satisfied, 
is impossible if $M \oplus N$ contains a subgraph on $5$ vertices of the form shown in Figure \ref{dense} (i). On the other hand, we can partition
the edges of this subgraph into an alternating cycle $C=(a,b,d,c)$, which is amenable w.r.t. both $M$ and $N$ and an alternating path $P=(x,d,a,x,b,c)$, which is amenable w.r.t. $M$. but not w.r.t. $N$. If, however, we reorder the edges of $P$ into $P'=(x,a,d,x,b,c)$ we obtain an alternating path amenable w.r.t $N$.

Let $H=(V_H,E_H)$ be an induced subgraph of $M \oplus N$. We say that it is {\bf \em dense} if  either  (i) $|V_H| =4$ and $|E_H|=6$ and $|E_H \cap M|=3$ or (ii) $|V_H| =5$ and $|E_H|\geq 9$ and $\lfloor\frac{|E_H|}{2} \rfloor \leq |E_H \cap M| \leq \lceil \frac{|E_H|}{2} \rceil$. We call such a subgraph a dense $4$-subgraph or a dense $5$-subgraph, respectively. Let us notice that each dense $4$-subgraph contains a unique alternating cycle.
In view of the above observations regarding Figure \ref{dense} we have:

\begin{fact}\label{fdense}
Let $M,N$ be two triangle-free $2$-matchings. The edges of any dense $5$-subgraph $H=(V_H,E_H)$ of $M \oplus N$ can be decomposed into $C$ and $P$ such that $C$ is an alternating cycle on $4$ vertices that is amenable w.r.t. both $M$ and $N$, $P$ is an alternating path or cycle that is amenable path w.r.t. $M$ and there exists a path or cycle $P'$ on the same set of edges as $P$ that is amenable w.r.t. $N$.
\end{fact}

Let $E'$ be any subset of edges of the graph. We say that $\sigma(E')=(v_0, v_1, \ldots, v_l)$ is an {\bf \em ordering} (of $E'$) if $|E'|=l$ and there exists
a bijection $o: \{1,2,\ldots, l\} \rightarrow E'$ such that for each $i: \ 1 \leq i \leq l$ it holds that $o(i)=(v_{i-1}, v_i)$.

\begin{figure}[h]
\centering{\includegraphics[scale=0.8]{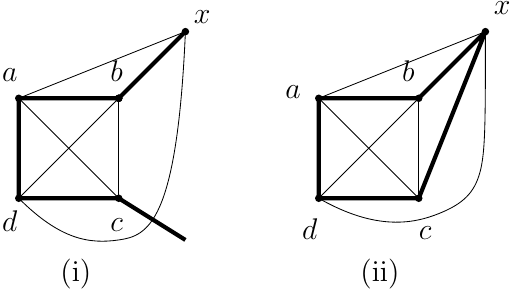}}
\caption{ Dense subgraphs on $a,b,c,d,x$.}
 \label{dense}
\end{figure}

The theorem that we prove is:

\begin{theorem} \label{theorem:decomposition}
Given two triangle-free  $2$-matchings $M_1$ and $M_2$, there exists a decomposition of $M_1 \oplus M_2$ into edge sets $E_1, E_2, \ldots, E_k$ such that
for each $E_i$ there exist orderings $\sigma_1$ and $\sigma_2$ of $E_i$ such that for each $j \in \{1,2\}$ it holds that  $\sigma_j(E_i)$ is an $M_j$-alternating path/cycle  amenable with respect to $M_j$. \\
Moreover, $\sigma_1(E_i) \neq \sigma_2(E_i)$ only if $|E_i|\geq 6$  and $E_i$ goes through a dense $5$-subgraph.
\end{theorem}

\dowod
Let us rename the matchings $M_1$ and $M_2$ as $M$ and $N$.  
We present the method of computing such a decomposition.  At the beginning the set $D$ of alternating cycles and paths is set to an empty set.
In the first stage we exhaustively find any dense $4$-subgraph $H$ of $M \oplus N$ and add an alternating cycle contained in it to $D$. If such a dense subgraph $H$
is contained in a dense $5$-subgraph $H'$ we order the remaining edges of $E_{H'}$ in two possible ways as indicated in Fact \ref{fdense}. If we obtain a path, then we continue extending it in the second stage. For every triangle $t \subseteq M \cup N$  such that $M \cap N \cap t \neq \emptyset$, we bind two edges of $t \setminus (M \cap N)$. It is easy to notice that no two such operations of binding are conflicting. This way we are done with all triangles of $M \cup N$ that are not contained in $M \oplus N$.


We now proceed to the description of the second stage. At each moment only one path $P$ (that may turn out to be a cycle) is under construction. We extend $P$ edge by edge (each edge is added to the endpoint of $P$) and may end the construction of $P$ only if both $M\oplus P$ and $N \oplus N$ are $2$-matchings. (Note that $P$ that is maximal (under inclusion) or an alternating cycle satisfies: both $M\oplus P$ and $N \oplus N$ are $2$-matchings.)  While extending $P$ we respect the previous binding operations. 
Apart from ensuring that the constructed paths and cycles are alternating, we apply the following two rules.

{\em Rule} $1$: Let $e$ be an edge contained in a triangle $t$ of $M \oplus N$  such that (i) either $t \cap M =\{e\}$ or $t \cap N =\{e\}$ (i.e., $e$ is the sole edge of $M$ or $N$ on $t$) and (ii) $e$ is not already bound with any edge of $t$. Then, whenever such $e$ is included in $P$, we bind it with another edge of $t$.
{\em Rule} $2$:  Whenever an edge $e$ belonging to a triangle $t$ of $M \oplus N$ is  both preceded and succeeded on $P$ by an edge not belonging to $t$, we bind the other two edges of $t$. 
We can notice that no two executions of Rule $1$ are conflicting, because if an edge $e$ is a sole edge of $M$ of two triangles $t_1, t_2$, then an edge $e'$ preceding $e$ on $P$ must belong to $t_1$ or $t_2$ and therefore we need to take care of only one of these triangles regarding Rule $1$.
Two executions of Rule $2$ can be conflicting only if $P$ contained a bumpy path on a dangerous subgraph. Below we show how to avoid bumpy paths.
Similarly, Rule $1$ cannot be carried out only if  it is conflicting with some execution of Rule $2$ and if some path or cycle in the decomposition contains a bumpy path. This follows from the following.
Suppose that at some point we have included an edge $e=(b,c) \in M$ such that the other two edges of $t=(a,b,c)$ belong to $N$, $c$ is the ending point of a currently constructed path $P$ and  the edge $(a,c)$ is already bound with some edge different from $e$. It means that $M \oplus N$ contains another triangle $t'=(a,c,d)$ that contains two edges of $M$. Since $(a,c)$ is bound with $(c,d)$, $D$ contains a path or cycle $P'=(\ldots, d,a,b,f, \ldots)$, where $f \neq c$. But this means that $(d,a,b)$ is a bumpy path in a dangerous subgraph that contains also the edges $(d,c), (a,c), (c,b)$.

We deal now with bumpy paths.
Let $S$ be a dangerous subgraph of $M \oplus N$ on vertices $a,b,c,d$ as described above.
First we notice that $M \oplus N$ cannot contain the edge $(a,c)$, which would necessarily have to belong to $N$, because then in the first stage we would have removed the alternating cycle $(a,b,d,c)$. Thus:
\begin{claim}\label{trbumpy}
 Any path $(u,v,w)$ contained in a triangle $t=(u,v,w)$ of $M \cup N$ is not bumpy.
\end{claim}

Suppose that we are in the middle of constructing an alternating path $P$ and we have ended  on $(c,b)$, which belongs to a bumpy path $(c,b,a)$ in a dangerous subgraph $S$ described above. Then we extend $P$ by $(b,f)$. We prove that  $(c,b,f)$ such that $(b,f) \in M$ and $f \neq a$ is not a bumpy path.
If $(c,b,f)$ is a bumpy path in a dangerous subgraph $S'$, then $S'$ must contain also either (1) $(b,d)$ or (2) $(b,a)$. In  case (1) $S'$ would also have to contain $(d,c)$ and $(d,f)$ and both would have to belong to $M$ but this is impossible because $d$ cannot be incident to three $M$-edges. In  case (2) $S'$ would have to contain the edge $(a,c)$ but we have observed above that this cannot be the case. If $(b,f)$ does not exist, we end $P$ on $(c,b)$. Suppose next that when constructing an alternating path $P$ we have ended  on $(a,b)\in S$. Then we can extend $P$ by $(b,d)$. Note that the path $(a,b,d)$ is not  bumpy by Claim \ref{trbumpy}.
 This means that we can always extend $P$ so that the added edge together with the preceding one are not a bumpy path.

How do we choose the starting edge of a path? If there exists an unfinished  path $P$ constructed during the first stage or a pair of bound edges, also denoted $P$, unassigned to any path or cycle of $D$, we extend $P$. Otherwise, we start from any edge. Let us finally point out the following: (i) if during construction of $P$,  we apply Rule $2$, then it may happen that a pair of edges $e_1, e_2$ that becomes bound is such that $e_1$ is the beginning edge of $P$ (which has just gone through $e_3$ such that $e_1, e_2, e_3$ form a triangle) and, (ii) if at a certain point $P$ could form an alternating cycle, but by binding the beginning edge $e_1$ of $P$ with the ending one $e_2$, we would violate Rule $1$ or $2$, we continue extending $P$ and (thus do not form that cycle.)

\koniec

\begin{corollary}\label{decompcore}
Let $M$ be a triangle-free $2$-matching, which is not maximum. Then there exists an $M$-augmenting path $P$ that is amenable with respect to $M$. \\

A symmetric difference of two triangle-free $2$-matchings $M_1, M_2$ can be decomposed into a collection  of alternating paths and alternating cycles such that each of the paths and cycles is feasible with respect to both $M_1$ and $M_2$.

\end{corollary}

\section{Outline}

Given an undirected simple graph $G=(V,E)$, a triangle-free $2$-matching $M$ and an unsaturated vertex $s$, we would like to find an augmenting  path $P$ starting at $s$ such that $M \oplus P$ is again triangle-free. From the previous section we know that we may restrict ourselves to augmenting amenable paths. In Figure \ref{trg3} we have three graphs such that in each of them the only unsaturated vertices are $s$ and $t$. Each of these graphs contains an augmenting path between $s$ and $t$. They are, respectively, $(s,d,b,c,e,t), \ (s,f,b,a,d,e,a,c,g,t),$ \\ 
$(s,h,a,b,e,d,b,c,g,f,c,a,i,t)$.  However, none of them is feasible, because  an application of any one of them to $M$ creates a triangle $(a,b,c)$. 

The algorithm for finding an augmenting amenable path starting at $s$ builds a subgraph $S$ of $G$ containing  paths beginning at $s$. In this respect it is similar to the algorithm for finding an augmenting path in the problem of finding a maximum matching. Let $P_v$ denote a  path staring at $s$ and ending at $v$. To the subgraph  $S$ we successively add edges and vertices in such a way that for any even length  path $P_v$ in $S$ it holds that $S$ contains  an even length amenable path $P^A_v$. Whenever we would like to augment $S$ by extending a path $P^A_v$ with the edges $e=(u,v)\notin M$ and $e'=(v,w_1)\in M$ but  the path $P^A_v \cup \{e,e'\}$ is non-amenable,  we forbid $e$ to be followed by $e'$ (in what sense we explain further) and do not add $e'$ to $S$ but may add $e$ and another edge  $e''=(v,w_2) \in M$.  It may happen that at some later point we discover in $G$ an amenable path such that $e$ and $e'$ are consecutive on it and then we extend $S$ accordingly. For example, in Figure \ref{trgtype2} at the  point when $S$ contains the edges $(s,d), (d,b)$, we extend it by adding $(b,c), (c,a)$ and $(b,f), (f,g)$. Later  when $S$ contains also $(g,a),(a,b)$  we will be able to add also $(c,e)$.

\begin{figure}[h]
\centering{\includegraphics[scale=0.6]{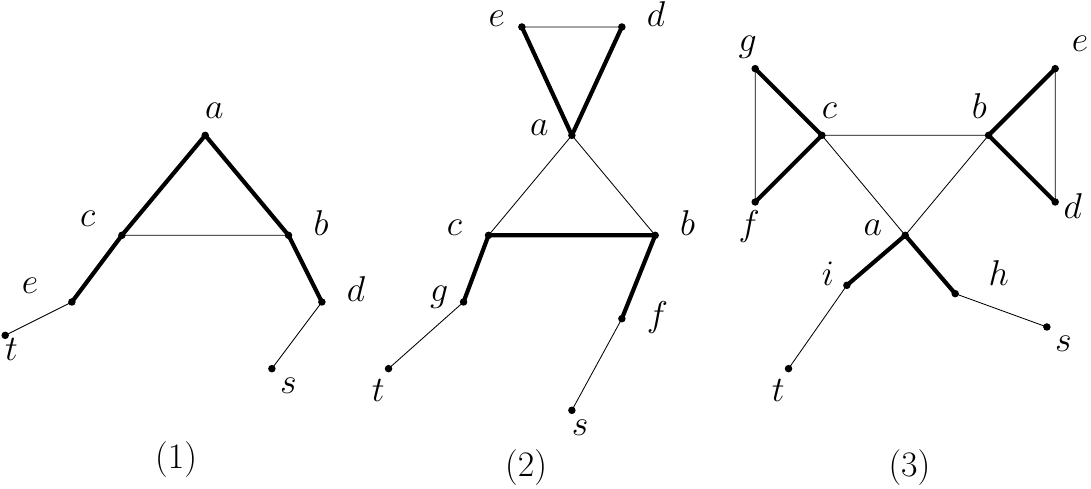}}
\caption{\small In each case the application to $M$ of an augmenting path  from $s$ to $t$  creates a triangle $t=(a,b,c)$.}
 \label{trg3}
\end{figure}

\begin{figure}
\centering
\begin{subfigure}{0.3\textwidth}
\includegraphics[scale=0.6, width=\textwidth]{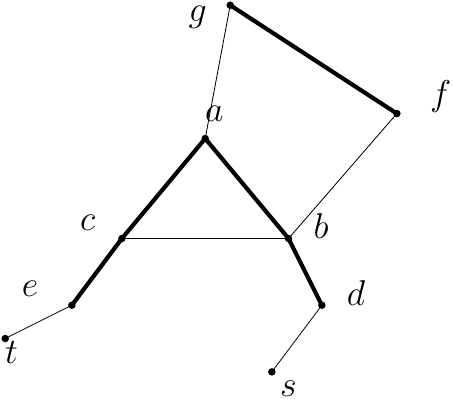}
\caption{}
 \label{trgtype2}
\end{subfigure}
\begin{subfigure}{0.3\textwidth}
\includegraphics[scale=0.6, width=\textwidth]{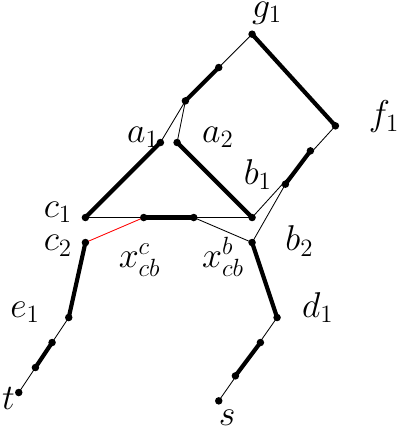}
\caption{}
\label{trgtype2half}
\end{subfigure}
\caption{\small At the point when we have not checked adding  the edge $(g,a)$ yet (it might also be the point when we have already added the edge $(a,g)$ but haven't added $(f,g)$ or $(f,b)$), the triangle $t=(a,b,c)$ is vulnerable and we remove the hinge $h(b,c,e)=(x^c_{cb},c_2)$ from $G_2$. After the addition of $(a,g)$ and $(a,b)$ to $S$ the triangle $t$ ceases to be vulnerable, because $S$ contains a path $(s,d,b,f,g,a,b,c,e)$ that is amenable on $t$ and hence $h(b,c,e)$ is passable and we restore it in $G_2$.}
\label{triangle2wersje}
\end{figure}

\vspace{2cm}

To be able to carry out the operation of forbidding some edges to be followed by certain others,
the algorithm works on $G'=(V' \cup X, E')$ obtained from $G$ by splitting vertices and replacing each edge not belonging to $M$ with a  path of length $3$. More precisely, we proceed as follows. For each edge $(u,v) \notin M$ we add in $G'$ two new vertices $x^u_{uv}, x^v_{uv} \in X$ and three edges: $(u, x^u_{uv}), (x^u_{uv}, x^v_{uv}), (x^v_{uv},v)$. The edge $(x^u_{uv}, x^v_{uv})$ is added to $M'$.
 Next for each vertex $v \in G$ we add  two new vertices $v_1, v_2 \in V'$, called copies of $v$.   We replace each edge $(v, x^v_{vw})$ with two edges $(v_1, x^v_{vw})$ and  $(v_2, x^v_{vw})$, each of which is called  a {\bf \em half-edge}.  We also form a matching $M'$ in $G'$.  For each edge $(u,v)\in M$ we replace it with one edge of $M'$ of the form $(u_i,v_j)$, where $i,j \in \{1,2\}$. We carry out the replacements so that as a result for each vertex $v$ of $G$ each of its two copies in $G'$ is matched in $M'$ to at most one vertex. Observe  that each edge $(u,v) \in M$ is replaced with one edge in $G'$ and each edge $(u,v) \in E \setminus M$ with five.
  Let $M'$ (resp.$M$) denote a matching of $G'$ (resp.$G$) and let $s$ be a fixed vertex of $G$ (and $G'$)  unsaturated in $M$ (unmatched in $M'$). At the beginning of the algorithm we define a graph $G_2$ equal to $G'$. We then modify $G_2$ during the algorithm as we build the subgraph $S$ and search for an amenable  path from $s$ to an unsaturated vertex in $G$.

We classify triangles of $M$ according to the  number of contained edges of $M$. We say that a triangle $t$ of $G$ is {\bf \em of type $i$} if it has exactly $i$ edges of $M$. Clearly, $i$ can only belong to the set $\{0,1,2\}$.
Let $t=(a,b,c)$ denote a triangle of type $1$ or $2$. If $t$ is of type $1$ and $(b,c) \in M$ or of $t$ is of type $2$ and $(b,c) \notin M$, then we say that $a$ is {\bf \em the top vertex} of $t$. A vertex $a$ is defined as a top vertex of a  triangle $t=(a,b,c)$ of type $0$ if among all vertices of $t$ it is the first such that $S$ contains some edge $(a,v) \notin M$ and $a$ was added to $S$ before $v$. (If such vertex does not exist, then $t$ has no top vertex.) By writing $t=(a;b,c)$ we denote that $a$ is the top vertex of a triangle $t$. \\

Let $e=(b,c), e'=(c,d)$ be two edges sharing an endpoint $c$ such that for some triangle $t$ we have $e \in t\setminus M, e' \in M \setminus t$. Then a {\bf \em a hinge (of $t$)} denoted as  $h(e,e')$ or  $h(b,c,d)$ is defined as that half-edge of $(b,c)$, which connects $x^c_{bc}$ to the edge $(c,d)$, i.e.,   $(x^c_{cb}, c_i)$, where  $c_i$ denotes that copy of $c$, which is $M'$-matched to $d_j$. A hinge $h(b,c,d)$ is said to be incident to $c$.  Let $t=(a;b,c)$ be a triangle of type $1$ or $2$ and $(b,d), (c,f)$ edges of $M\setminus t$.  Two {\bf \em basic} hinges of $t$ as defined as follows: $h(a,b,d)$ and $h(a,c,f)$ if $t$ is of type $1$ and $h(b,c,f)$ and $h(c,b,d)$ if $t$ is of type $2$. See Figure \ref{hinges}.
Hinges play an important role in characterising non-amenable paths:

\begin{observation} \label{zawobs}
Let $t=(a;b,c)$ be  any triangle of $G$ and $(b,d), (c,f)$ edges of $M\setminus t$. (See Figure \ref{hinges}.) Then a  path $P$ of $G_2$ is non-amenable on $t$ if and only if:
\begin{enumerate}
\item $P$ contains both basic hinges of $t$, if $t$ is of type $1$ or $2$;
\item $P$ contains six hinges of $t$, if $t$ is of type $0$.

\end{enumerate}

\end{observation}
\dowod It follows from the fact that a path non-amenable on $t$ of type $i$ must contain all $i$ edges of $t \setminus M$, none of which is (immediately) preceded or followed by an edge of $t \cap M$.
\koniec

\begin{figure}[h]
\centering{\includegraphics[scale=0.8]{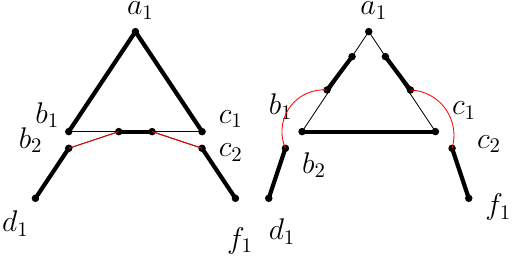}}
\caption{Basic hinges are drawn in red}
 \label{hinges}
\end{figure}

The operation of forbidding an edge $e$ to be followed by $e'$ consists in removing the hinge $h(e,e')$ (i.e., an appropriate half-edge) from $G_2$.
We say that a path $P$ in $G_2$ goes through or contains the edge $(u,v) \notin M$ of $G$ if $P$ contains the edge $(x^u_{uv}, x^v_{uv})$.

\subsection{Structure $S$ and graph $G_2$}

Finding an $M$-augmenting path in $G$ starting at $s$ may be reduced to finding  an $M'$-augmenting path in $G'$ starting at $s$.
To make the presentation clearer, we extend $G_2$ and $G'$ so that $s$ is the only unmatched vertex in $M'$  in the following manner.
For every unmatched in $M'$ vertex $v_i$ different from $s$ we add a new vertex $v^M_i$ and a new edge $(v_i, v^M_i)$ to $M'$. The set of all such new vertices is denoted as $V_s$.
Thus we get an obvious correspondence between $M$-augmenting paths in $G$ starting in $s$ and even-length  paths beginning in $s$ in $G'$ and ending in vertices of $V_s$.
\begin{observation}
An even length  path from $s$ to $v^M_i$ in $G'$ corresponds to an $M$-augmenting path from $s$ to $v$ in $G$.
\end{observation}

The algorithm maintains a structure $S$ that at the beginning consists only of the vertex $s$. Further on, a vertex $v$
is contained in $S$ only if $S$ contains a path from $s$ to $v$.
Let $\mathcal{E}(S)$ denote the set $\{v \in S: S$ contains an even-length  path from $s$ to $v \}$.  Vertices of $\mathcal{E}(S)$ are called {\bf \em even}.
Let $\mathcal{O}(S)$ denote the set $\{v \in S: S$ contains an odd-length path from $s$ to $v \}$. 
If a vertex $v$ belongs to $\mathcal{O}(S) \cap \mathcal{E}(S)$, then $S$ contains a blossom $B$, defined as follows.
\begin{definition}(\cite{Duan})
A blossom is identified with a vertex set $B$ and an edge set $E_B$ on $B$. If $v$ is a vertex of $G'$, then $B=\{v\}$ is a singleton blosssom with $E_B=\emptyset$.
Suppose there is an odd-length sequence of vertex-disjoint blossoms $B_0, B_1, \ldots, B_k$ with associated edge sets $E_{B_0}, \ldots, E_{B_k}$. If $\{B_i\}$ are connected in  a cycle by edges $e_0, e_1, \ldots, e_k$, where $e_i \in B_i \times B_{i+1}$ (modulo $k+1$) and $e_1, e-3, \ldots, e_{k-1}$ belong to $M'$, then $B= \sum_i {B_i}$ is also a blossom associated with edge set $E_B=\sum_i {E_{B_i}} \cup \{e_0, e_1, \ldots, e_k\}$. Blossoms $B_0, \ldots, B_k$ are called subblossoms of $B$ and blossom $B_0$ - the {\bf \em base subblossom} of $B$. The one vertex $b$ in $B$, which is left unmatched in $M' \cap E_B$, is called the {\bf \em base} of $B$. $S$ contains  an even-length  path $P$ from $s$ to a vertex $b$ such that $P$ and $B$ are vertex-disjoint except for $b$.
\end{definition}
 A blossom $B$ that is not contained in any other blossom is called {\bf \em outer}. For a vertex $v$ by $B(v)$ we denote the outer blossom containing $v$. If $v$ does not belong to any blossom, then $B(v)$ denotes a singleton blossom consisting only of $v$.
In many algorithms dealing with similar structures, a blossom is  treated as one vertex and shrunk. The algorithm in the paper also maintains $\bar{S}$, the subgraph $S$ with blossoms contracted. $\bar{S}$ is a tree. It is also an ''alternating tree''. Any path from a vertex in $S$ to $s$ is alternating. The definitions of structure $S$ and subgraph $\bar{S}$ used here are  the same as in \cite{Gabow17} Section 3.2

For every vertex $v \in \mathcal{E}(S)$, we associate with $v$  one even-length  path $P_v$ contained in $S$ starting at $s$ and ending at $v$. Generally, we want $P_v$ to satisfy the following property:
If $P_v$ uses at least one edge of a blossom $B$, then all edges of $P_v \cap B$ form one subpath $P'$ of $P_v$ such that
 one endpoint of $P'$ is the base of $B$.  (In some cases, we will allow $P_v$ to slightly violate this property - compare Figure \ref{strukturaS}.)

 An edge $(u,v) \in E'\setminus M'$ is $S$-augmenting if $B(u) \in S, v \notin B(u)$ and either (i) $v \notin S$ or (ii) $v \in \mathcal{E}(S)$. After the addition of such an edge to $S$, $S$ either contains a new vertex $v$ or $S$ contains a new blossom, which is formed by a fundamental cycle $C$ in $\bar{S} \cup (u,v)$.
For an edge $e=(u, v) \in E' \setminus M'$, let  $e^+= e \cup \{(u, M'(u)), (v, M'(v))\}$. 
Let $ S^+$ denote all edges of $G_2$ with both endpoints in $S$. For any two vertices $u_i, v_j$ such that $(u,v) \notin M$, we say that a path $(u_i, x^u_{uv}, x^v_{uv}, v_j)$ contained in $G_2$ is a {\bf \em segment} of $G_2$. We denote it as $\sigma=(u_i,v_j)$.

Our goal is to compute a set of $M$-augmenting paths starting at  $s$ instead of all vertices of $V' \cup V_s$ reachable from $s$ via an even-length  path.
In the case of $2$-matchings, we can  notice that if $v_i \in \mathcal{E}(S)$, then we may sometimes ignore extensions to $S$, which would result only in that the other copy  of $v$ would belong to $\mathcal{E}(S)$. 
For this reason we introduce the notion of  alternating $s$-reachability.
We say that two subgraphs $H_1, H_2$ of $G_2$ have the same {\bf \em alternating $s$-reachability} if for every vertex $v \in V \cup V_s$ it holds that $H_1$ contains an even length  path from $s$ to some copy of $v$ if and only if $H_2$ does. 

\begin{definition}



An $S$-augmenting hinge $e=h(u,v,w)$ is {\bf \em $S$-impassable} if    for some vertex $z \in \mathcal{E}(S\cup e^+)$ a  path $P_{z}$  is non-amenable on a  triangle $t$ containing $(u,v)$.


An $S$-augmenting hinge $e=h(u,v,w)$ is $(S^+)$-{\bf \em impassable} if  there exists a vertex $z \in \mathcal{E}(S\cup e^+)$ such that each even-length path from $s$ to $z$ in $S^+ \cup e^+$ is non-amenable on a  triangle $t$ containing $(u,v)$.


A hinge that is not   impassable is {\bf \em passable}.


An $S$-augmenting  hinge $h(e,e')=h(u,v,w)=(v_j, x^v_{uv})$  is {\bf \em safe} if (i) $(v,w)$ belongs to a triangle of $G$, (ii) $\mathcal{E}(S)$ already contains some copies of $v$ and $w$ and (iii)  some copies of $v$ and $w$ already belong to a common blossom or the addition of $h(e,e')$ would result in creating a new blossom $B'$, whose only new edge of $M'$ is $(v,w)$.

An $S$-augmenting hinge $h(u,v,w)$ is {\bf \em vulnerable} if  it is impassable (on a triangle $t$) and  not safe.  We call such a hinge a {\bf \em vulnerable} hinge of $t$ and denote  as $h(t)$.  If some hinge of a triangle $t$ is vulnerable, then $t$ is also called vulnerable.

\end{definition}

Suppose that $G_2$ is as shown in Figure \ref{trg3} (2) and that the edges of $M$  correspond to the following  edges of $M'$:  $(f_1, b_2), (b_1, c_1), (c_2,g_1), (a_1, d_1), (a_2, e_1)$. Let us consider the  point when $S$ contains a path $(s,f,b,a,d,e,a)$ and we are thinking of extending $S$ by a hinge $h(c,a,e)=(x^a_{ac}, a_2)$. Notice that at this point $S$ also contains a blossom $B=(x^a_{ab}, a_1, d_1, x^d_{ed}, x^e_{ed}, e_1, a_2, x^a_{ab})$. A path from $s$ to $x^c_{ac}$ in $S \cup h^+(c,a,e)$ is amenable and hence we extend $S$ by $h^+(c,a,e)=(x^a_{ac}, a_2) \cup (x^a_{ac}, x^c_{ac})$. Next, if we would like to extend $S$ by $h(a,c,g)$, we notice that a path in $S \cup h^+(a,c,g)$ from $s$ to $g_1$ is non-amenable on $(a;b,c)$ and so is every path in $S^+ \cup h^+(a,c,g)$ from $s$ to $g_1$, which means that $h(a,c,g)$ is impassable.
 In Figure \ref{trg3} (1) the hinge $h(b,c,e)$ is impassable and in Figure \ref{trg3} (3) the hinges $h(b,a,i)$ and $h(c,a,i)$ are impassable.

It may happen for a hinge $h(e,e')$    that an even-length path in $S \cup h^+(e,e')$ from $s$ to $w_k$ is non-amenable  but  $ S^+ \cup h^+(e,e')$ contains an amenable even-length path from $s$ to $w_k$. Consider an example from Figure \ref{triangle1overlap} (a). Suppose that we first add to $S$ the edges $(s,b), (b,b')$. Next, we add  pairs of edges $(b',x) \notin M, (x, M(x))$:
$(b,a), (a,a'), (a, a''), (b,d), (d,d'), (d,d'')$.
Further, we add also $(a', a'')$ and form a blossom $B_a$ with a base in $x^a_{ab}$. When we consider edges incident to $B_a$ we notice a vulnerable hinge $h(a,c,f)$ 
(the path $(s,b',b,a,a',a'',a,c,f)$ is non-amenable on $t=(a,b,c)$). We also add $(a,c), (c,b)$ to $S$. Further, when we consider $d'$, we add $(d',d'')$ and form a blossom $B_d$. When we check if we can extend $S$ by $(c,d)$, we notice that $S \cup (c,d), (c,f)$
 contains a path non-amenable on $t'=(d,c,b)$. But $ S^+\cup (c,d), (c,f)$ contains a path $(s,b',b,b,a,a',a'',c,b,d,d', d'',d,c,f)$, which is amenable. This path is not contained in $S$, because $S$ does not contain the hinge $h(d,b,c)$ (instead, $S$ contains $h(b',b,c)$.) It will turn out, however, that there are only three cases, when we  need to consider $ S^+$ instead of $S$ to notice that a hinge is not impassable: (i) described above when $t=(a;b,c), t'=(d,b,c)$ are of type $1$, (ii) when $t=(a;b,c)$ is of type $1$ and $t'=(b;c,d)$ is of type $2$ and (iii) when $t$ is toppy.

\section{Algorithm in more detail}

\subsection{Main ideas}
The pseudocode of an algorithm  computing an amenable augmenting paths starting at  $s$ is given below as Algorithm \ref{algnew}. Strictly speaking, Algorithm \ref{algnew} computes a set $S$ of amenable paths starting at $s$ such that for each vertex $u$ of $V \cup V_s$, $S$ contains an even-length path from $s$ to $u$ if and only if $G$ contains a feasible even-length alternating path from $s$ to $u$.
Special blossoms used in this algorithm are defined later. For now, one can think of  Algorithm \ref{algnew} as though it did not contain lines \ref{spec1} and \ref{spec2}. Let us now outline the main ideas behind the algorithm.

\begin{algorithm}[H]
\begin{algorithmic}[1]

	\State Let $s$ be an unmatched vertex in $M'$. Set $S \leftarrow \{s\}$.  
	Assume that each vertex $v$ is contained in a blossom containing only $v$.

	\While {$\exists$ an  $S$-augmenting  edge $e=(w,w') \in G_2 \setminus E_{sf}$  with $w \in \mathcal{E}(S)$} 
	
	\While  {$\exists$ an  $S$-augmenting  edge $e=(u_i,x^u_{uv}) \in G_2 \setminus E_{sf}$ incident to $B(w)$} \label{whileproc}
	\State \textsc{AugmentS$(u_i,x^u_{uv})$}
	\If{$x^v_{uv} \in \mathcal{E}(S)\setminus B(w)$}
	
	    \If {$(x^v_{uv},v_1) \in G_2 \setminus E_{sf}$ is  $S$-augmenting}
			   \State \textsc{AugmentS$(x^u_{uv},v_1)$}
			\EndIf
			\If {$(x^v_{uv},v_2) \in G_2 \setminus E_{sf}$ is  $S$-augmenting}
			   \State \textsc{AugmentS$(x^u_{uv},v_2)$}
			\EndIf
			
  \EndIf			
				
	\EndWhile
	\EndWhile
	
	\vspace{0.5cm}

	\Procedure{AugmentS}{$e=(u,v)$}
	
	\If {$e$ is impassable}
	  \If {$e$ is vulnerable}
		\State remove $e$ from $G_2$
		\Else 
		 \State add $e$ to $E_{sf}$
	  \EndIf
		
	\Else 
	  \If {$u, v \in \mathcal{E}(S)$}
		 \State \label{sig} add $e$ to $S$
		 \State \label{blos} form a new blossom $B$, which contains  $e$
		 
		   \If  {$B$ is  special} \label{spec1}
	     \State \label{spec2} remove hinge(s) $h(B)$ defined in Lemmas \ref{speciallem} and \ref{special0} from $S$ 
	 
		   \EndIf
	  \Else 
	     \State \label{sig+} add $e^+$ to $S$
			 
	  \EndIf
	
	 \State \label{vul} restore every  hinge $e_h$ that is not vulnerable 
	\EndIf
	
\EndProcedure
	
\end{algorithmic}
\caption{Computing a set of amenable paths starting at an unsaturated vertex $s$.}
\label{algnew}
\end{algorithm}

We want to extend the structure $S$ in a similar way as we would be extending it with respect to a $2$-matching $M$ in the original graph $G$. For this reason, (i) we
try to extend $S$ by segments, rather than only by edges of $G_2$ (that is why an extension by $(u_i,x^u_{uv})$ is followed by extensions by $(x^v_{uv},v_1), (x^v_{uv},v_2)$ whenever possible) and (ii) as long as a given  blossom $B(w)$ has an incident $S$-augmenting edge, we add first such edges to $S$, before proceeding to edges incident to another blossom.  

We say that a triangle $t=(a;b,c)$ of type $1$ or $2$ is {\bf \em toppy} if $ S^+$ contains a path $P$ from $s$ to $a$ that contains no edge of $t$.
We say that a triangle $t=(a;b,c)$ of type $0$  is {\bf \em toppy} if $ S^+$ contains an odd-length path $P$ from $s$ to $a$ that does not contain all edges of $t$.

If we built the structure $S$ in the graph  $G_2$ without any removed hinges, then a path non-amenable on any triangle could take only one of the three forms depicted in Figure \ref{trg3}. This is because whenever we process a blossom, we add all incident $S$-augmenting edges. Therefore, for example, if $G$ contains a triangle $t=(a;b,c)$ of type $1$ shown in Figure \ref{trg3} (2), and at some point the structure $S$ contains a path $(s,f,b,a,d)$, then at some later point $S$ could not contain a path of the form  $P'=(s,f,b,a,d,u,v,g,c,a,e)$, because after adding the edges $(b,a)$ and $(a,d)$ to $S$ (while processing $b$), we would also add $(a,e)$.  In $G_2$ with some hinges removed, a path non-amenable on a triangle of type $0$ or $1$ could potentially have a more complex form, for example, as in $P'$ - it can happen if a hinge $h(b,a,e)$ is removed in $G_2$.
In the first two lemmas below we show that it does not happen and   that the only forms of non-amenable paths that we are going to encounter in $S$ in the course of the algorithm still have the form  depicted in Figure \ref{trg3}. 

In the algorithm we must be able to recognize if $ S^+$ contains an even-length amenable path from $s$ to a given vertex $z$ or not. This is needed to establish if a hinge is vulnerable (and hence impassable) or if it ceased to be one. Suppose that we are examining a hinge $h(u,v,w)$.  The structure $S \cup h^+(u,v,w)$ contains only one  even-length  path $P_z$ from $s$ to $z$. We check if this path is amenable on a triangle $t_{uv}$ containing $(u,v)$. If it is, then  the hinge $h(u,v,w)$ is passable. Otherwise, we examine if $S^+$ contains another ''roundabout'' even-length path from $s$ to $z$ that does not contain both basic hinges of $t_{uv}$ (if $t$ is of type $1$ or $2$) or six hinges of $t_{uv}$ (if $t$ is of type $0$). It turns out that to do this, it suffices to verify only two things:
(i) if $t_{uv}$ is toppy and (ii) if there exists another triangle $t'$ sharing one edge with $t_{uv}$ such that in the algorithm, $t'$ has previously  been recognized as vulnerable. If any of these two points holds, then the hinge $h(u,v,w)$ is not vulnerable (and hence $ S^+$ contains an even-length amenable path from $s$ to a given vertex $z$ or $h(u,v,w)$ does not increase the alternating $s$-reachability of $S$). A high-level reason for this is the following. 

Suppose that  $t=(a;b,c)$ is a triangle of type $1$, see Fig.\ref{hinges}, and  to reach a vertex $z$ we have to go through the hinge $h(a,b,d)$. We want to check if $S^+ \cup h^+(a,b,d)$ contains an even length path $P_z$  that is amenable on $t$, i.e.  $P$ that  does not contain the other basic hinge $h(a,c,f)$ of $t$.   One possibility of the existence of such a path  is when $t$ is toppy. The other possibility is a path of the form $P_2=(s, \ldots, b', b, c,a, \ldots, a, b, d)$. In the second case, it might seem that if $S$ contains a path $P'=(s, \ldots, b', b)$ , then to get to $d$ we could extend this path $P'$ by going directly from $b$ to $d$ and not going through edges of $t$. However, such an extension might denote a path non-amenable on some $t'$ and therefore mean that $P_2$ is indispensable.

Similarly, if $t$ is a triangle of type $2$, then a path $P$ that contains one basic hinge $h(b,c,f)$ of $t$ is amenable on $t$  only if it does not contain the other basic hinge $h(c,b,d)$ of $t$. Hence, such $P$ should contain $(a,b)$ and  not contain $(d,b)$. One possibility of the existence of such a path is when $t$ is toppy. The other possibility is a path of the form $(s, \ldots, f', c, a, \ldots, a, b, c, f)$. However, in the second case, when there exists an odd-length alternating path from $s$ to $c$ that does not go through $t$, we prove below that we do not really need the hinge $h(b,c,f)$ to reach the edge $(c,f)$. 

 In Lemmas \ref{triangle1} and \ref{toplem} we prove that a toppy triangle of type $1$ or $2$ cannot be vulnerable and in Lemmas \ref{overlaptr0}, \ref{overlaptr1},
\ref{overlaptr2} that any two vulnerable triangles are edge-disjoint. These provide an efficient way of verifying whether a given hinge is impassable or not.
To be able to quickly notice when a vulnerable triangle $t$ becomes toppy, the top vertex of $t$ needs to remain unburied in any blossom that contains some edge of $t$.
That is why we do not form so-called special blossoms - these are, roughly speaking, blossoms that contain at least one edge of a vulnerable triangle.

To {\bf \em process} a blossom $B(v)$ means to carry out the loop {\em while} from line \ref{whileproc} of Algorithm \ref{algnew}. 
 
\subsection{Technical part}

In the first two lemmas we prove properties of vulnerable triangles of type $1$ and $0$. Their proofs are in Section \ref{deferred}.

\begin{lemma} \label{triangle1}
Let $t=(a;b,c)$ be a vulnerable triangle of type $1$  and $(a,d), (a,e), (c',c), (b',b)$ edges of $M$ incident to $t$. Then the vulnerable hinge $h(t)$ of $t$ is its basic hinge. If  $h(t)=h(a,b,b')$, then $S$ contains a path $(s, \ldots, c',c, a,d, \ldots, e,a)$.  Moreover, $t$ is not toppy, there exists no vulnerable triangle $t' \neq t$ sharing an edge with $t$ such that $h(t')$ is incident to $a$ and $S$ contains a blossom $B(a)$ with a base in $x^a_{ac}$ and with edges $(a,d), (a,e)$.
\end{lemma}

\begin{lemma} \label{triangle0top}
Let $t=(a;b,c)$ be a triangle of type $0$. If $S$ contains all edges of $t$, then they are contained in one blossom $B$ of the form $(a,b, \ldots, b,c, \ldots, c,a)$ with a base in $a$. If $t$ has an impassable hinge, then it is incident to $a$. If $t$ shares an edge with a triangle $t'=(c';a,c)$ of type $2$, then $t$ is not vulnerable.
\end{lemma}


\begin{figure}[h]
\centering{\includegraphics[scale=0.7]{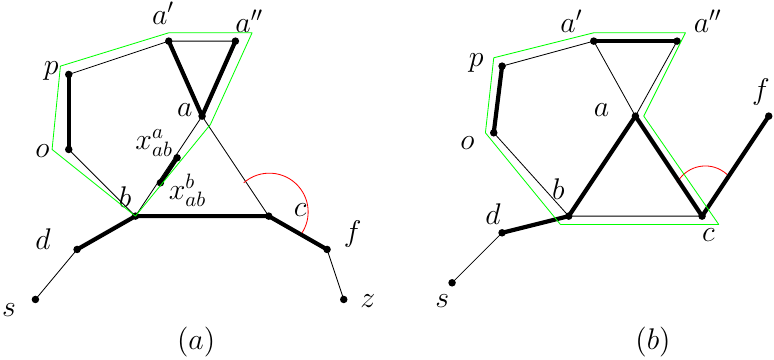}}
\caption{\small In both cases $S$ contains a subblossom $B_1=(a,a',a'',a)$. $S$ also contains a blossom $B=(b, B_1, p, o, b)$ in case (a) and a blossom $B=(b,c,B_1,p,o,b)$ in case (b).  In case (a) the subblossom $B_1$ has its base in $x^a_{ba}$ and  an amenable path from $s$ to $f$ has the form $(s,d,b,o,p,a',a,x^a_{ac}, x^c_{ac},c,f)$. (Note that it does not go through the base of $B_1$). Assume that in case (b) the triangle $t=(a;b,c)$ consists of edges $(a_1,b_1), (a_2, c_1)$ and $seg(b_1, c_1)$ in $G_2$. Then the base of subblossom $B_1$ is $a_2$
and $S$ contains even length paths $P_{a_2}=(s,d,b_2,c,a_2)$ and $P_{a_1}=(s,d,b_2,o,p,a',a'',a_2,c,b_1,a_1)$. Since $S$ also contains a path $P_{b_1}=(s,d,b_2,c,a,a'',a',a,b_1)$, the blossom $B$ is extended to the blossom $B_t$ containing all edges of $t$. An amenable path from $s$ to $f$ has the form $(s,d,b,o,p,a',a'',a_1,b_1,c,f)$ - it does not go through the base of $B_1$ and hence violates slightly the property required of paths $P_v$.
} \label{strukturaS} 
\end{figure}


\begin{lemma}\label{toplem}
If a triangle $t=(a;b,c)$ of type $2$ is toppy and there is no unprocessed blossom $B(v) \in \mathcal{E}(S)$ such that $v \in t$, then $t$ is not vulnerable.


\end{lemma}
\dowod
 Suppose that $M'$ contains edges $(a_1,b_1), (a_2, c_1), (b_2,d_1), (c_2,f_1)$. Let us consider the first   step in the algorithm (lines \ref{sig} and \ref{sig+}), after which $t$ becomes toppy.  After this step $S$  contains a path $P=(s,\ldots, x,a,c)$ (or $P=(s, \ldots, f,a,b)$). (Note, that  we extend $S$ by  $e^+$ and not only by $e$. Therefore, for example while adding $h(x,a,c)$ we would also add $(a_2,c_1)$ if it did not belong to $S$ already.) 

We can notice that in the case when none of the hinges $h(x,a,b), h(c,b,a), h(b,c,a)$ turns out to be impassable, the algorithm will extend $S$ so that it contains also the edges $(a,b), (b,c),h(c,b,a), h(b,c,a),(x^c_{bc}, x^b_{bc})$ (the last three edges form a segment $(b_1,c_1)$) and as a result $S^+ \cup h^+(c,b,d)$ will contain a  path $(s, \ldots, x,a,c,b,d)$ that clearly contains $h(c,b,d)$ and does not contain the other basic hinge of $t$ - $h(b,c,f)$. Note that a hinge $h(c,b,d)$ can be impassable only on $t$. Therefore, $h(c,b,d)$ will be included into $S$. A symmetric argument applies to $h(b,c,f)$. This shows that in this case, none of the basic hinges  of $t$ is vulnerable. 

Let us consider now the cases when some of the hinges $h(x,a,b), h(c,b,a), h(b,c,a)$ are impassable. We are going to show that in each of these cases $S$ will contain a path ending with $(x,a,c,b,d)$ (or a symmetric one ending with $(x,a,b,c,f)$) and from this point on, one basic hinge $h(c,b,d)$ of $t$ is clearly contained in $S$ 
and the other one $h(b,c,f)$ is either (i) not $S$-augmenting and at no further point will it become one  or (ii)  even if it  becomes  $S$-augmenting, it will never become vulnerable, because it will be safe.  

The hinge $h(c,b,a)$ may be impassable because of a triangle $t_1=(c;b,d)$ of type $1$,  a triangle $t_2=(d;b,c)$ of type $2$ or  a triangle $t_3=(b;c,i)$ of type $0$, see Figure \ref{zawiasywewfig}. If $h(c,b,a)$ is impassable on $t_3$, then $S$ contains a path $(s, \ldots, d,b,i,j, \ldots, g,i,c,f,\ldots,a,c,b)$, which means that
$S$ contains the hinge $h(c,b,d)$ and the other hinge $h(b,c,f)$ of $t$ is not $S$-augmenting. 
If $h(c,b,a)$ is impassable on $t_1$, then the algorithm realizes this when $S$ contains a path $P=(s, \ldots, i,d,c,f,\ldots, a,c)$ and a segment $\sigma=(c_1, b_1)$ is examined during processing of the blossom  $B(c_1)$ (in fact, by Lemma \ref{triangle1} $B(c_1)$ contains also $c_2$ and its base is $(x^c_{dc})$). During processing of $B(c_1)$ a segment $\sigma=(c_1, b_2)$ is also going to be examined and as a result $S$  will be extended by $\sigma^+=(c_1, b_2)^+$, because clearly a path $P'=(s, \ldots, i,d,c,f,\ldots, a,c,b,d)$ is amenable. (Note that it might happen that the algorithm added the hinge $h(b,c,f)=(x^c_{bc}, c_2)$ while processing $B(c_2)=B(c_1)$ and then $S \cup h(c,b,d)$  contains a path $(s, \ldots, d,c,a, \ldots, f,c,b,d)$ which is not amenable on $t$. However, $S^+ \cup h(c,b,d)$ contains also a path $(s, \ldots, d,c,f, \ldots, a,c,b,d)$ which is amenable on $t$. Although, it is not necessary, in this case we can replace $h(b,c,f)$ with $h(b,c,a)$ in $S$.) This means that the hinge $h(b,c,d)$ is included in $S$. As for the hinge $h(b,c,f)$, we can notice that at no point  will it become $S$-augmenting. If $h(c,b,a)$ is impassable on $t_2$, then the algorithm realizes this when $S$ contains a path $P=(s, \ldots, a,c)$ and a segment $\sigma=(c_1, b_1)$ is examined during processing of the blossom  $B(c_1)$. When the algorithm examines $\sigma=h(c,b,d)$, $S$ will be extended by $h(c,b,d), (b,d)$ and hence a path $P'=(s,\ldots,a,c,b,d)$ will be contained in $S$. From this time on, the hinge $h(c,b,d)$ belongs to $S$. If the hinge  $h(b,c,d)$  becomes $S$-augmenting at some later point, $x^c_{bc}$ will have to belong to $\mathcal{E}(S)$, which means that $b_2$ also belongs to $\mathcal{E}(S)$. And if $b_2 \in \mathcal{E}(S)$, then it means that $b_2$ belongs to a blossom $B(b_2)$ that contains edges $(b,d), (b,c), (c,a)$.  Hence, at such a point the hinge $h(b,c,d)$ will be safe (and will be unable to become vulnerable.)


The hinge $h(x,a,b)$ may be impassable because of a triangle $t_1=(x;a,c)$ of type $1$ or because of a triangle $t_2=(c;a,f)$ of type $2$, see Figure \ref{zawiasytopfig}. (If hinge $h(x,a,b)$ were impassable on a triangle of type $0$, then $(a;b,c)$ could not be toppy.)  By the above, we can also assume that none of the hinges $h(b,c,a), h(c,b,a)$ is vulnerable and hence none of them is removed. 
In the first case $S$ contains a path $P_1=(s, \ldots, f,c,x, \ldots, x, a,c)$ and thus $c_1 \in \mathcal{E}(S)$. This means that $\sigma=(c_1, b_2)$ is examined during 
processing of $B(c_1)$. As a result $\sigma^+=(c_1, b_2)^+$ is added to $S$ and thus the hinge $h(c,b,d)$ is contained in $S$. We also observe that the hinge $h(b,c,f)$ will not become $S$-augmenting at any later point of the algorithm, because of the following. The hinge $h(b,c,f)$ could become vulnerable only if  $x^c_{bc}\mathcal{E}(S)$ and this could happen if $b_2 \in \mathcal{E}(S)$, which in turn implies that the edge $(b,d)$  would have to belong to blossom $B(b_2)$.  However, any blossom $B(b_2)$ also contains the edges $(x,a), (a,c), (c,b)$, which means that $h(b,c,f)$ is not $S$-augmenting.  In the case when $h(x,a,b)$ is impassable because of a triangle $t_2=(c;a,f)$ of type $2$ (the hinge $h(x,a,b)$ denotes then $h(f,a,b)$), $S$ contains a path $P_2=(s, \ldots, i,f,a,c)$,  which will be extended along the edges $(c,b)$ and $(b,d)$ during processing of $B(c_1)$ and examining the segment $\sigma=(c_1, b_2)$.  As a result, the hinge $h(c,b,d)$ is included in $S$. The algorithm may examine the other hinge $h(b,c,f)$
when processing a blossom $B(b_2)$.  Such $B(b_2)$ will then contain also the edges $(b,c), (c,a), (a,f)$, see Figure \ref{2triangles2}. It may happen that $h(b,c,f)$ is impassable but it will never become vulnerable because it will be safe. \koniec

\begin{figure}[h]
\centering{\includegraphics[scale=0.7]{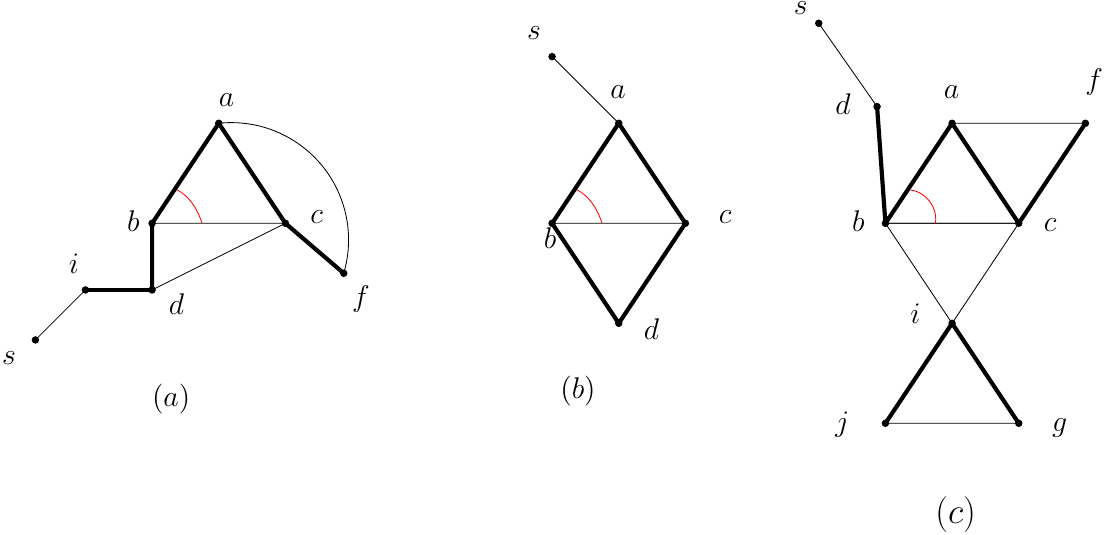}}
\caption{\small In all cases the hinge $h(c,b,a)$ is impassable.
} \label{zawiasywewfig}
\end{figure}

\begin{figure}[h]
\centering{\includegraphics[scale=0.7]{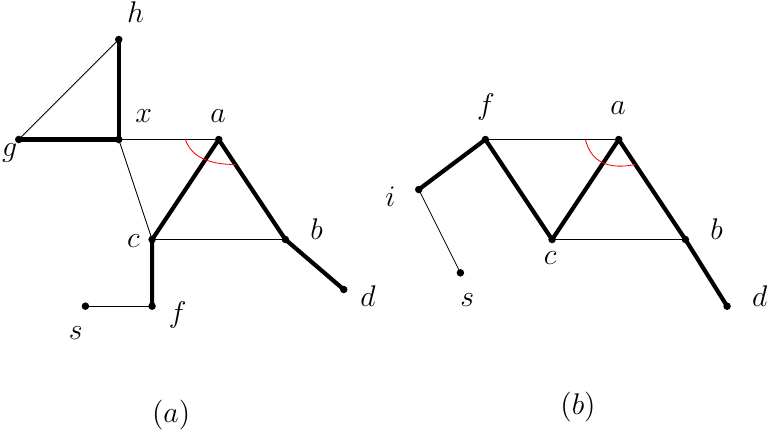}}
\caption{\small In case (a) hinge $h(x,a,b)$ is impassable and in case (b) hinge $h(f,a,b)$ is impassable. } 
\label{zawiasytopfig}
\end{figure}

\begin{figure}
\centering{\includegraphics[scale=0.8]{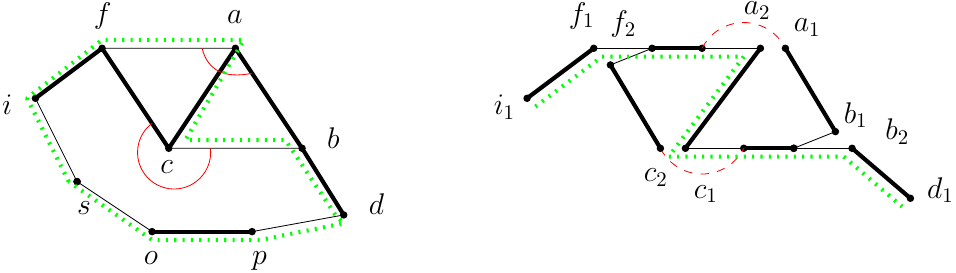}}
\caption{The dotted edges form a blossom $(s, o, p, d,b,c,a,f,i,s)$. The hinges $h(b,c,f), h(f,a,b)$ are impassable but safe, therefore not vulnerable. Safe hinges are proved to never increase the alternating $s$-reachability.
} \label{2triangles2}
\end{figure}

In the following three lemmas we prove that vulnerable triangles are edge-disjoint throughout the execution of the algorithm.

\begin{lemma}\label{overlaptr0}
Let $t=(a;b,c)$ be a triangle of type $0$ that has an impassable hinge at some point and shares an edge with a triangle $t' \neq t$.
Then $t'$ is not vulnerable at this or at any later point of the algorithm.
\end{lemma}
\dowod
 Since $t$ has an impassable hinge, $S$ contains a path $P=(s, \ldots, f,a,c,c', \ldots, c'',c,b,b', \ldots, b'',b,a)$ which does not contain $(a,d)$.
Suppose that $t'$ is of type $1$, see Figure \ref{triangle0overlap}. In all such cases $t'$ is toppy and hence cannot be vulnerable.

Suppose next that $t'$ is a triangle of type $0$. Then the only possibility for $t'$ to have an impassable hinge is when $t$ and $t'$ have common top vertices. Hence, $t'=(a;b,d)$. (The case when $t'=(a;c,d)$ is symmetric.) If $t'$ has a vulnerable hinge, then $S$ contains a path $P'=(s, \ldots, f,a,d,d', \ldots, d'',d,b,b', \ldots, b'',b,a)$. But the existence of $P$ and $P'$
in $S$ implies that $S$ also contains a path \\
$P''=(s, \ldots,  f,a,c,c', \ldots, c'',c,b,b', \ldots, b'',b, d, d', \ldots, d'',d, a)$, which proves that both $t$ and $t'$ is toppy. A toppy triangle of type $0$ has by  definition no vulnerable hinge.

The  case  when $t'$ is of type $2$ has already been considered in Lemma \ref{triangle0top}.
\koniec

\begin{figure}[h]
\centering{\includegraphics[scale=0.7]{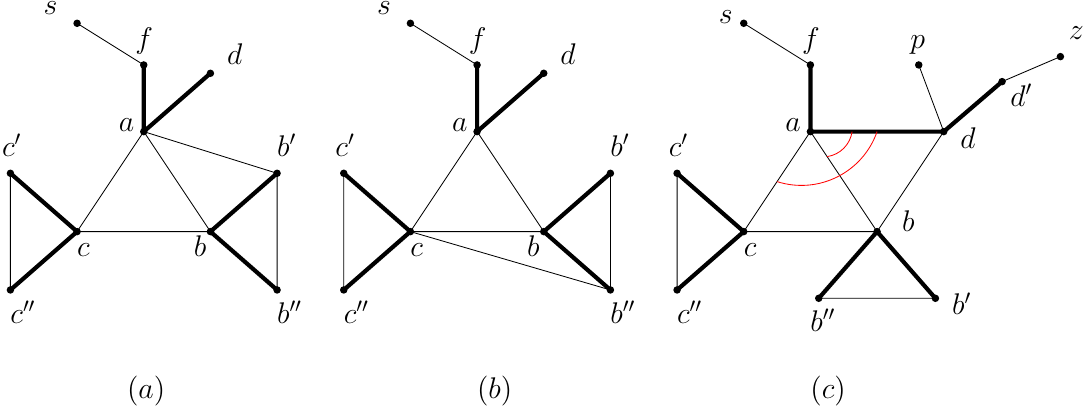}}
\caption{}
\label{triangle0overlap}
\end{figure}

\begin{observation}\label{zawiasy}
Let $t=(a;b,c)$ be a triangle of type $2$ and $(c,f) \in M, (f',c) \notin M$ edges outside of $t$. Then $h(f',c,f)$ is not a hinge of any triangle of type $2$.
\end{observation}
\dowod
If $h(f',c,f)$ were a hinge of a triangle $t'$ of type $2$, then $t'$ would have to have the form $(a;c,f')$, but since $f' \neq b$, such $t'$ cannot occur in $G$.
\koniec

\begin{lemma}\label{overlaptr1}
Let $t$ be a triangle of type $1$ that shares an edge with a triangle $t' \neq t$ of type $1$ or $2$. Suppose that $t$ is vulnerable.
Then $t'$ is not vulnerable at this or at any later point of the algorithm.
\end{lemma}

\dowod Suppose that $t=(x;a,c)$ and $t'=(a;b,c)$ is of type $2$ and that $(b,d), (c,f), (x,g), (x,h)$ are edges of $M \setminus (t \cup t')$, see Figure \ref{zawiasytopfig}(a). If $t$ is vulnerable, then it has an impassable hinge and it is either $h(x,a,b)$ or $h=(x,a,c)$ and hence, $S$ contains one of the paths: $P_1= (s, \ldots, f,c,x,g, \ldots, h,x,a)$ or $P_2=(s, \ldots, b',b,a,x,h, \ldots, g,x, c)$. If $S$ contains $P_1$, then $t'$ is toppy and by Lemma \ref{toplem}  not vulnerable. Let us now assume that $S$ contains $P_2$ and that the edges $(b,d), (c,f)$ correspond to the edges $(b_2, d_2), (c_2, f_2)$ of $M'$. 
Before $t'$ can become vulnerable, $\sigma=(b_2, c_2)$ has to be $S$-augmenting. Let us notice that since the edge $(b_1, a_1)$ is contained in $S$, the edge $(b_2, d_2)$ is also in $S$ and $b_2 \in \mathcal{O}(S)$. This is so because $(b_1, a_1)$ was added while examining some segment $\sigma=(u_i, b_1)$ during processing of $B(u_i)$ 
and segment $\sigma=(u_i, b_2)$ would have been considered then too. The hinge $h(u_i,b_2, d_2)$ is not a hinge of any triangle of type $2$ by Observation \ref{zawiasy}.
We can also assume that $t$ is the first vulnerable triangle of type $1$ that shares an edge with $t'$, hence $h(u_i,b_2, d_2)$ is also not a vulnerable hinge of a triangle of type $1$. Clearly, it cannot be a vulnerable hinge of a triangle of type $0$ either. This means that indeed $(b_2, d_2) \in S$.
Once $b_2 \in \mathcal{E}(S)$, segment $\sigma=(b_2, c_1)$ will be examined and $\sigma^+=(b_2, c_1)^+$ will be added to $S$. After processing $B(a_1)$ (which is a singleton), $S^+$ contains an amenable path $(s, \ldots, d,b,c,a,x , \ldots, x,c,f)$. Note that the hinge $(a_1, x^a_{ax})$ is not contained in $S$. Therefore $t'$ ceases to be vulnerable (and $t$ does not become one).

Suppose next that $t=(c;b,d)$ and $t'=(a;b,c)$ is of type $2$, see Figure \ref{zawiasywewfig}(a). 
 The  case when the basic hinge $h(c,b,a)$ of $t$ is vulnerable has already been considered in Lemma \ref{toplem}. The second when the second basic hinge $h(c,d,i)$ of $t$ is vulnerable means that $t'$ is toppy and by the same lemma, not vulnerable.

Let us consider next the cases when $t=(a;b,c)$ and $t'$ is of type $1$. Suppose first that $t'=(d;b,c)$. See Figure \ref{triangle1overlap}(a). Let 
$(b,b'), (f,c), (a,a'), (a,a''), (d,d'), (d,d'')$ denote edges of $M \setminus (t \cup t')$.
Since $t$ has an impassable hinge, then  $S$ contains exactly one of the paths
$P_1=(s, \ldots, b',b,a,a', \ldots, a'',a,c)$ or $P_2=(s, \ldots, f,c,a,a', \ldots, a'', b)$ (the order of $a', a''$ is arbitrary - $P_2$ might also have the form: $(s, \ldots, f,c,a,a'', \ldots, a', b)$). Assume that it is $P_1$.  If $t'$ is to have an impassable hinge too, then  $S$ would have to contain a path $P'=(s, \ldots, b',b,d,d', \ldots, d'',d,c)$. (Note that $t'$ cannot have an impassable hinge because of  a path $P''=(s, \ldots, f,c,d,d'',d',b)$.) By combining these two paths via $(b,c)$ we get an amenable path $(s, \ldots, b',b, a,a', \ldots, a'',c,b,d,d', \ldots, d'', d,c,f)$ that contains the hinge $h(d,c,f)$. This shows that if $S$ contains both $P_1$ and $P'$, then neither $t$ nor $t'$ is vulnerable.

Suppose next that  $t'=(a;c,f)$ and let $(b,b'), (f,o), (a,a'), (a,a'')$ denote edges of $M \setminus (t \cup t')$. See Figure \ref{triangle1overlap}(b). Suppose that $S$ contains $P_1$. Then $t'$ cannot have an impassable hinge, because  $S$  contains a path $(s, b', b, a, a', \ldots, a'', a, f,o)$.
The cases when $t'=(b;a,a')$ or when $t'=(c;a,a')$ have already been considered in Lemma \ref{triangle1}.
\koniec

\begin{figure}[h]
\centering{\includegraphics[scale=0.7]{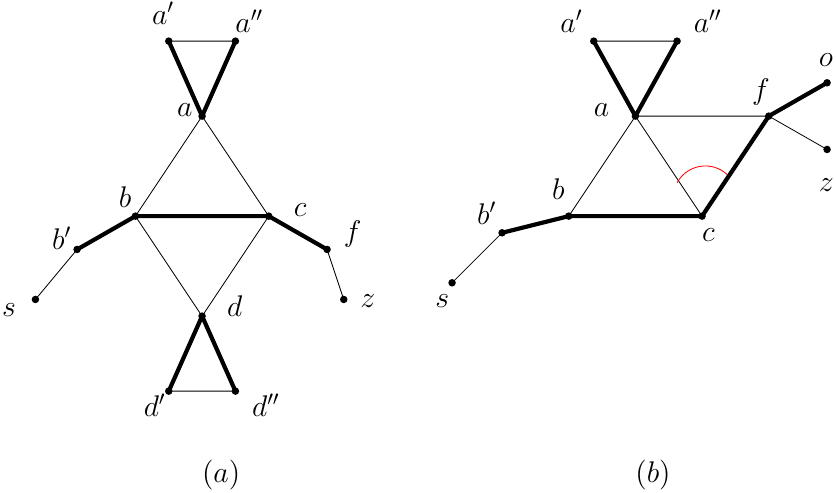}}
\caption{}
\label{triangle1overlap}
\end{figure}



\begin{lemma}\label{overlaptr2}
Let $t$ be a triangle of type $2$ that shares an edge with a triangle $t' \neq t$ of type  $2$. Suppose that $t$ is vulnerable at some point.
Then $t'$ is not vulnerable at this or at any later point of the algorithm.

\end{lemma}

This follows from the proof of Lemma \ref{toplem}.

How can we efficiently recognize that a triangle $t=(a;b,c)$ of type $1$ or $2$ is toppy? We can observe that if $S$ contains an odd/even alternating length path from $s$ to $a$, then $a$ belongs to $\mathcal{O}(S)$ or $\mathcal{E}(S)$, respectively. On the other hand, the top vertex $a$ of $t$ may belong to corr. $\mathcal{O}(S)$ or $\mathcal{E}(S)$ when $t$ is not toppy - such cases happen when $a$ is contained in some non-singleton blossom. Let $B$ denote the most outer blossom containing $a$. Then it holds, that if $B$ does not contain any edge of $t$, then $t$ is toppy only if $B$ (as a shrunk vertex) belongs to appropriately $\mathcal{O}(S)$ and $\mathcal{E}(S)$. If $B$ contains at least one edge of $t$,
then $a$ belongs to $\mathcal{E}(S) \cap \mathcal{O}(S)$ (because each vertex of a blossom belongs to $\mathcal{E}(S) \cap \mathcal{O}(S)$) but $t$ sometimes is toppy and sometimes not.
It turns out, however, that a blossom $B$ that contains at least one edge of a vulnerable triangle $t$ of type $1$ or $2$ has a rather special form described below in Lemma \ref{speciallem}.

\begin{lemma}\label{speciallem}
Let $t=(a;b,c)$ be a vulnerable triangle of type $1$ or $2$ with   $h(t)=h(a,c,f)$  for $t$ of type $1$ and $h(t)=h(b,c,f)$ for $t$ of type $2$. Let $B$ be    a blossom containing at least one edge of $t$.  Then $B$ consists of all edges of $t$. Moreover, the base of $B$ is $B(b)$. 
\end{lemma}

\dowod
Let $(c,f), (b,d)$ denote edges of $M \setminus t$ which correspond to edges $(c_2,f_i), (b_2, d_j)$ of $M'$. Assume also that the counterparts in $M'$ of edges $(a,b),(a,c)$ are the edges $(a_1, b_1), (a_2, c_1)$.

 Suppose first that $t$ is of type $2$.  The blossom $B$ cannot contain only the edge $(b,c)$ of the edges of $t$, because then it would have to go through the hinge $h(b,c,f)$, which is not present in $G_2$. $B$ may not contain only one of the edges $(a,b), (a,c)$ (and additionally, possibly $(b,c$) because then $t$ would have to be toppy and hence by Lemma \ref{toplem} not vulnerable. If, on the other hand, $B$ contains both edges $(a,b), (a,c)$ but not the edge $(b,c)$, then $t$ cannot be vulnerable. This follows
from the fact that such $B$ has form $(x, \ldots, d',b,a, \ldots, a,c, f', \ldots, x)$ and by Lemma \ref{zawiasy} the hinges $h(d',b,d)$ and $h(f',c,f)$ are not impassable. Therefore, we can extend $S$ so that it includes the edges $(b,d), (c,f)$ and the hinges $h(d',b,d), h(f',c,f)$, which means that none of the hinges $h(c,b,d), h(b,c,f)$ is $S$-augmenting. Also, at no further point of the algorithm will any of these two hinges be $S$-augmenting, because at the point when  both $b_2$
and $c_2$ belong to $\mathcal{E}(S)$ (this is the only possibility for $seg(b_2,c_2)$ to be $S$-augmenting), the edges $(d,b), (c,f)$ are contained in $B$, therefore $seg(b_2,c_2)$ is not $S$-augmenting.

Hence, the only form $B$ may take is: $B=(x, \ldots, d', b,a, v, \ldots, v', a, c,b,d, \ldots, x)$, where $x$ is the base of $B$. We claim that the beginning part of $B$ from $x$ to $b$ going through $(d',b)$ and the last part starting from the edge $(b,d)$ and ending on $x$ form a subblossom $B_1$ - this follows from the fact that $(d,b)$ and $(d',b)$ cannot  form an impassable hinge $h(d',b,d)$ by Lemma \ref{zawiasy}. On the other hand, the part of $B$ from $a$ to $a$ forms another blossom $B_2$ (whose base is either $a_1$ or $a_2$, depending on which of them was processed first), which follows from the fact none of the hinges $h(v,a,b), h(v',a,c)$ can be vulnerable because it would mean that there exists a vulnerable triangle $t'$ sharing an edge with $t$. (Note, however, that it may happen that there exists a vulnerable triangle $t'=(a;u, u')$  of type $0$ but then either hinges $h(u,a,b),  h(u',a,b)$ are vulnerable or hinges $h(u,a,c), h(u',a,c)$ and not for example, $h(u,a,b), h(u',a,c)$, which would prevent forming a blossom $B_2$.)

The reasoning for a triangle of type $1$ is analogous.
\koniec

A blossom described in Lemma \ref{speciallem} is called {\bf \em special}. An example of a special blossom is depicted in Figure \ref{specialblossomex}.

\begin{figure}[h]
\centering{\includegraphics[scale=0.7]{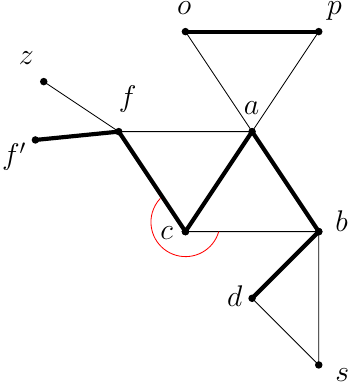}}
\caption{\small A special blossom $B$ is formed by edges of $t=(a;b,c)$, edges of $t_1=(s;b,d)$ forming a subblossom $B_1$ and edges of $t=(a;o,p)$ forming a subblossom $B_2$. The base of $B_2$ is either $a_1$ or $a_2$, depending on which of them was processed first.  $B_1$ is the base of $B$.}
\label{specialblossomex}
\end{figure}

It turns out  that we do not have to form special blossoms in $S$ to preserve its properties.

\begin{lemma} \label{special}
Let $B$ denote a special blossom in $S$ on $t=(a;b,c)$ of type $1$ or $2$, whose base is a blossom containing $b$ and  let  $(b,d)$ and $(c,f)$ denote edges in $M \setminus t$ and $(b,d')$ an edge not in $M$.  Let $h(B)$ denote the following hinge:

\begin{enumerate}
\item If $t$ is of type $2$, then $h(B)=h(c,b,d)$ unless $G$ contains a triangle $t'=(c;a,f)$ of type $2$, in which case $h(B)=h(d',b,a)$.
\item If $t$ is of type $1$, then $h(B)=h(a,b,d)$.

\end{enumerate}

Then $h(B)$ does not increase the alternating $s$-reachability at any further point and is not needed as long as $t$ is vulnerable.

Also, for any vertex $v \in V' \cap \mathcal{E}(S)$ it holds that if $S$ contains an amenable path $P_{v'}$, then so does $S \setminus h(B)$.
\end{lemma}
\dowod
Assume first that $t$ is a triangle of type $2$. Suppose that the edges of $M' \cap t$ are $(a_1, b_1), (a_2, c_1)$. If $G$ does not contain $t'=(c;a,f)$ of type $2$, then we remove $h(c,b,d)$ from $S$ (but not $S^+$) and if  $bs(B(a))=a_2$, we change $bs(B(a))$ to $a_1$. The only thing that changes, apart from the disappearance of the blossom $B$, is that now $a_2 \notin \mathcal{E}(S)$. If $G$ does contain a triangle $t'=(c;a,f)$ of type $2$ (see Figure \ref{specialblossomex}), then we remove $h(d',b,a)$ from $S$ and if  $bs(B(a))=a_1$, we change $bs(B(a))$ to $a_2$. After these modifications $a_1 \notin \mathcal{E}(S)$. The reason for why we remove $h(d',b,a)$ in this case and not $h(c,b,d)$ is that if we removed $h(c,b,d,)$, then $t'$ would stop being toppy and $S$ would not contain an amenable path to $f'$, where $(f,f')$ is an edge of $M \setminus t'$.


Suppose now that $t$ is a triangle of type $1$.  Suppose that the edge of $M' \cap t$  $ (b_1, c_1)$.We remove $h(a,b,d)$ from $S$ (but not $S^+$) and if  $bs(B(a))=x^a_{ac}$, we change $bs(B(a))$ to $x^a_{ab}$. After these modifications $b_1 \notin \mathcal{E}(S)$ but $b_2$ still belongs to  $\mathcal{E}(S)$.

{\em Remark: } If $t$ is of type $2$ and $h(B)=h(c,b,d)$, then we may need $h(B)$ later, when the algorithm discovers a triangle $t'=(f;a,c)$ of type $1$, whose basic hinge $h(f,c,b)$ is $S$-impassable  - since $t$ and $t'$ share an edge, we know by Lemma \ref{overlaptr1} that they cannot be  vulnerable. 

\koniec

We assume that we do not form special blossoms in $S$ as described in Algorithm \ref{algnew}.


In Lemma \ref{triangle0top} we have observed that the edges of a vulnerable triangle $t=(a;b,c)$ of type $0$ are contained in one blossom $B$. We also call a blossom $B$ {\bf \em special} (on $t$).

We have the analogue of Lemma \ref{special}:

\begin{lemma}\label{special0}
Let $B$ be a special blossom on a vulnerable triangle $t=(a;b,c)$ of type $0$ such that $h(t) =h(b,a,e)$. Let $(b_1,b'_1), (b_2, b'_2) \in M'$. After the removal of hinges $h(c,b,b'_1)$ and $h(c,b,b'_2)$ from $S$ (but not from $G_2$), the alternating $s$-reachability of $S$ does not change. The hinges $h(c,b,b'_1)$ and $h(c,b,b'_2)$ are  denoted as $h_1(B), h_2(B)$.
\end{lemma}
The proof is similar to the proof of Lemma \ref{special}.

\begin{lemma} \label{safe}
A safe hinge never increases the alternating $s$-reachability.
\end{lemma}
The proof is given in Section \ref{deferred}.

\section{Correctness and running time} \label{corr}
While examining whether a hinge $h(u,v,w)$ is impassable, we only verify if its addition does not create a path non-amenable on $t$ containing $(u,v)$. Hence, such verification is very local and there hypothetically exists a risk that this hinge is part of a path non-amenable on a triangle $t'$ that does not contain $(u,v)$.
Below we prove that all paths $P_v$ associated with the structure $S$ are amenable (on every triangle of $G$.)

\begin{lemma}
If a vertex $v$ belongs to $\mathcal{E}(S)$, then $S$ contains an amenable even length path $P^A_v$ from $s$ to $v$. 

\end{lemma}

\dowod
The only possibility for an even length alternating path $P_v$ contained in $S$ to be non-amenable would be if there existed a triangle $t$ in $G$ and $S$ contained two  paths \\
$P_ {v_1}, P_{v_2}$ such that $M \oplus (P_{v_1} \cup P_{v_2})$ is non-amenable on $t$. Note that these paths would have to end with paths $P'_1, P'_2$ (where $P'_i\ $ ends with $v_i$) such that $P'_1$ and $P'_2$ start and end with edges of $M$, are edge-disjoint and contain all edges
of $t\setminus M$.

Obviously such a triangle $t$ would have to be of type $1$ or $0$.
However, by Lemmas \ref{triangle0top} and  \ref{triangle1}, it is impossible. Note that by Lemma \ref{triangle0top}, any two paths containing different sets of edges of $t$
share at least one edge of $M$ - the one incident to the base $a$ of the blossom containing the edges of $t$. Similarly,  for a triangle $t=(a;b,c)$ of type $1$ it holds by Lemma \ref{triangle1} that any two paths containing edges $(a,b)$ and $(a,c)$ respectively and not containing $(b,c)$ share a matching edge incident to $a$.
Also, while reorganizing blossoms and hence attaching paths to a different copy of top vertex of a triangle, we might spoil amenability of some paths. However, 
this is prevented by two choices of reorganization, i.e., by two choices of a hinge $h(B)$.
\koniec

\begin{figure}[h]
\centering{\includegraphics[scale=1]{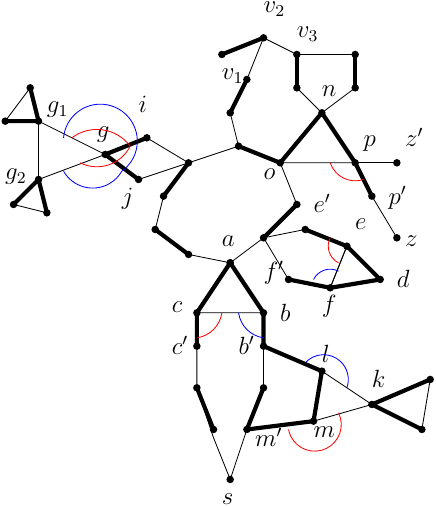}}
\caption{Vulnerable triangles: $t_1=(a;b,c), \ t_2=(d;e,f), \ t_3=(n;o,p), \ t_4=(k;l,m), \ t_5=(g;g_1,g_2)$. The vulnerable hinges of these triangles are drawn in red. Hinges drawn in blue are alternative possibilities of vulnerable hinges, i.e., $G_2$ contains either red hinges and then blue ones are vulnerable or the other way around. Triangle $t_5$ has two vulnerable hinges $h(g_1,g,i)$ and $h(g_2,g,i)$, or alternatively: $h(g_1,g,j)$ and $h(g_2,g,j)$. 
The graph $G$ contains two augmenting paths from $s$  - one infeasible  to $z$ and the other amenable to $z'$.}
 \label{bigexample}
\end{figure}

\begin{theorem}
If $G$ contains an amenable augmenting path  starting at $s$, then at the end of Algorithm \ref{algnew} structure $S$ contains a path starting at $s$ and ending in a vertex $v \in V_s$. 
\end{theorem}
\dowod
Suppose to the contrary that $G$ contains an amenable augmenting path $P$ starting at $s$ but $S$ contains no  path starting at $s$ and ending in a vertex in $V_s$.
This means that $ S^+$ also does not contain any alternating path starting at $s$ and ending in $V_s$.
Let $P$ be an amenable augmenting path from $s$ to $v \in V_s$ contained in $\bar{G'}$. 

\begin{claim} W.l.o.g. we may assume that if $P$ goes through a vulnerable  hinge of a  triangle $t$ of type $1$ or $2$, then $P$  goes through $t$  only once. 
\end{claim}
\dowod 
Recall that a path is said to go through an edge $(u,v) \notin M$ if it contains $(x^u_{uv}, x^v_{uv})$.
Let $t=(a;b,c)$ be a triangle of type $1$ or $2$ such that $P$ goes through a basic hinge of $t$ and goes through $t$ twice. Then $P$ has one of the forms:
(i) $(s, \ldots, b',b,a, a',\ldots, b'', b,c,a, a'', \ldots)$, (ii) $(s, \ldots, b',b,a,a', \ldots, a'', a,c,b, v, \ldots)$ or symmetric ones. Then $P'$ that has the form
(i) $(s, \ldots, b',b,a, a'', \ldots)$, (ii) $(s, \ldots, b',b,a, v, \ldots)$, correspondingly, goes through $t$ at most once and is amenable on each vulnerable triangle.
\koniec

For a triangle $t$ we define a {\bf \em pole} of $t$ as follows. If $t$ is of type $1$ or $2$, then $t$ has only one pole - its top vertex. If $t$ is of type $0$, then $t$ has two poles - each non-top vertex of $t$ is its pole.


\begin{claim}\label{claimcorr1}
Let $p_1$ be a pole of a vulnerable triangle $t_1$  and  $p_2$ a pole  of a vulnerable triangle $t_2 \neq t_1$. Then $p_1 \neq p_2$ and in $G_2$ there is no path $P$ from $p_1$ to $p_2$ that does not go through any edge of $t_1 \cup t_2$.
\end{claim}
\dowod 
From Lemmas proving that any two vulnerable triangles are edge-disjoint we get that if both $t_1$ and $t_2$ are of type $0$ or $1$, then $p_1 \neq p_2$. 
In $G$ there cannot exist two different triangles of type $2$ sharing the top (because then this top vertex would have to be incident to at least three edges of $M$.)
If $t_1$ is of type $2$ and $t_2$ is of type $0$ or $1$, then they cannot share a pole either, because if $t_2$ is of type $1$, then by Lemma \ref{triangle1}, the edges of $M \cap t_1$ would belong to a blossom (and no edge of a vulnerable triangle is contained in a blossom) and if $t_2$ is of type $0$, then it would imply that $t_1$ is toppy.

Suppose to the contrary that  a path $P$ from $p_1$ to $p_2$ from the claim is contained in $G_2$. 
The existence of such path from $p_1$ to $p_2$ would mean that $G_2$ contains an edge with both endpoints in $\mathcal{E}(S)$, which is neither a safe edge nor a special edge and hence that Algorithm \ref{algmain} failed to form some blossom (which is neither a special blossom nor a blossom arising only due to a safe edge) - a contradiction. \koniec

Recall that by Lemmas \ref{speciallem} and \ref{special0} no edge belonging to a vulnerable triangle belongs to a blossom of $S$.

If $P$ does not go through any hinge that is removed in $G_2$, it is wholly contained in $ S^+$. Since $S$ has the same alternating $s$-reachability as $ S^+$, it follows that $S$ contains a path from $s$ to $v$. Suppose then that $P$ goes through $k$  hinges  $h(t_1), \ldots, h(t_k)$, each of which is vulnerable - in this order, starting from $s$.  Each of these hinges is absent in $G_2$. The  part of $P$ from $s$ to $h(t_1)=h(u_1,v_1,w_1)$ (by definition $(u_1,v_1) \notin M$ and $(v_1,w_1) \in M$) is contained in $ S^+$. Note that $v_1$ is not a pole of $t_1$.
Since $t_1$ is vulnerable, it is not toppy in $G_2$. Therefore and because $P$ goes through $t$ of type $1$ or $2$ only once, the hinge $h(t_1)$ is traversed on $P$ in the order $w_1, v_1, u_1$ and not $u_1, v_1, w_1$. Hence, $w_1 \notin t_1$.   Let $P_1$ be the maximal part of $P$  starting at $v_1$ and containing  an edge of $t_1$.  Then, because $P$ is amenable, $P_1$ ends on a vertex  $v'_1$  that is a pole of $t_1$. (If $t_1$ is a triangle of type $1$, then already $u_1$ is a pole of $t_1$ and  $(v_1, u_1) \notin t_1$.
If $t_1$ is a triangle of type $2$, then $P$ must contain a subpath of the form $(w_1, v_1, u_1, v'_1)$, where $(u_1, v'_1) \in M \cap t_1$ - otherwise it would be non-amenable on $t_1$. Also, the vertex following $v'_1$ does not belong to $t_2$, because $P$ is $M$-alternating.)
The next time $P$ encounters a vulnerable hinge $h(t_2)=h(u_2,v_2,w_2)$, this hinge is also traversed on $P$ in the order $w_2, v_2, u_2$. This follows from the fact that otherwise $P$ would contain a path between $v'_1$ (a pole of $t_1$) and a pole of $t_2$, that is entirely contained in $S^+$, contradicting  Claim \ref{claimcorr1}. Note that we also use here the property that any two vulnerable triangles are edge-disjoint. (Otherwise, a path from $h(t_i)$ to $h(t_{i+1})$ might  first go through a pole of $t_{i+1}$ and later through a pole of $t_i$ and  then $h(u_{i+1}, v_{i+1}, w_{i+1})$ might be traversed in order $u_{i+1}, v_{i+1}, w_{i+1}$. For example, if $t_i=(a;b,c)$ and $t_{i+1}=(b;a,d)$ were vulnerable triangles of type $2$, then the path $(x,c,b,a,y)$ goes through the pole $b$ of $t_2$ and then through the pole of $t_1$ and thus $t_2$ would be toppy in the graph containing a restored $h(b,c,x)$.) Continuing in this manner we obtain that the last part of $P$ is a path $P_{k+1}$ between a pole  of $t_k$ and a vertex $v\in V_s$, which is entirely contained in $ S^+$, contradicting the fact that $ S^+$ contains no  vertex of $V_s$. 
\koniec

({\em Remark: } To prove tha above theorem, we can also use a weaker  Lemma $13$ from version $7$ of this paper on ArXiv which states  that if a a triangle-free $2$-matching $M$  is not maximum, then there exists an $M$-augmenting path $P$ that is amenable on every vulnerable triangle $t$. The proof of this lemma is  shorter than the proof of Theorem \ref{theorem:decomposition}.)

Let $n$ and $m$ denote the number of, correspondingly, vertices and edges in the graph $G$.

\begin{lemma} \label{runtime}
The running time of Algorithm \ref{algnew} is $O(m)$.
\end{lemma}
\dowod
The running time of Algorithm \ref{algnew} is basically the same as that of finding an augmenting path in non-bipartite graphs. 
We use data structures from  \cite{GabowTarjan} and generally the analysis from \cite{Gabow17} goes through.

Additionally, we need to be able to decide if a hinge $h(u,v,w)$ is impassable or not. 
How do we check it?
Assume that $(u,v)$ belongs to at least one triangle. Otherwise, $h(u,v,w)$ is clearly passable. We check if $S \cup h^+(u,v,w)$ contains a path $P$ non-amenable on a triangle $t$ containing $(u,v)$.  Notice  that checking if a given path is amenable on $t$ containing $(u,v)$ is easy, because 
if $P$ is non-amenable on $t$, then the last one, two or three edges of $E \setminus M$ of $P$ in $\bar{S} \cup h^+(u,v,w)$ ($\bar{S} - \ S$ with shrunk blossoms) belong to $t \setminus M$, depending on the type of $t$, by Lemmas \ref{triangle1} and \ref{triangle0top}.
If $S \cup h^+(u,v,w)$ contains a path non-amenable on a triangle $t$, then we check if $t$ is toppy or if there exists a vulnerable triangle $t'$ 
that shares an edge with $t$. By Lemmas \ref{overlaptr0}, \ref{overlaptr1} it would mean that $h(u,v,w)$ is passable. In the remaining case, $h(u,v,w)$ is indeed impassable and we remove it from $G_2$.
Thus, time needed to establish if a hinge is passable is constant. 

Similarly, time needed to carry out line \ref{spec2} is constant - we only need to restore a vulnerable hinge of $t$ if $t$ is toppy or shares an edge with another vulnerable triangle $t'$. (If a hinge stopped being vulnerable because it is not $S$-augmenting any more, then we do not need to restore it explicitly.) Restoring a vulnerable hinge consists in re-adding it to $G_2$. 

\koniec

\begin{corollary}
A maximum triangle-free $2$-matching of $G$ can be computed in $O(nm)$ time.
\end{corollary}
\dowod
We start from any maximal (w.r.t. inclusion) triangle-free $2$-matching $M$. While there exists an unsaturated vertex $v$, we run Algorithm \ref{algnew} to find an amenable augmenting path $P$ starting at $v$. If such a path exists, we augment $M$ by applying $P$ to it. If there is no amenable augmenting path starting at $v$, then there never will be, so $v$ will not be an endpoint of any other augmenting amenable path. Thus, in each iteration we either augment $M$ or convince ourselves that
the degree of $v$ in $M$ will not change. Clearly, there can be only $O(n)$ such iterations.
\koniec

The above approach could be modified so that it runs in $O(|M_{opt}| m)$ time, where $M_{opt}$ denotes a maximum triangle-free $2$-matching.

\vspace{1cm}

{\em Remarks about the  structure $S$} \\

If a vulnerable triangle $t$ at some point becomes toppy or is discovered as not vulnerable anymore because it shares an edge with a triangle $t'$ such that $S$ contains a path non-amenable on $t'$, then we sometimes reorganize $S$ (and $\bar{S}$) accordingly. (These modifications are not necessary though.) For example, in the case shown in Figure \ref{triangle1overlap} (a), the modified $S$ would contain a path $(s,b',b,b,a,a',a'',c,b,d,d', d'',d,c,f)$. Similarly, in the case shown in Figure \ref{zawiasytopfig} (a) at one point $S$ contains a path
$P_1=(s, f,c,b,a)$ and the hinge $h(c,b,d)$ is removed. After adding a path $(s,f,c,x,g,h,x,a,c)$ to $S$, $t$ is discovered as toppy and then we remove the hinge $h(b,c,f)$ from $S$ (but not from $G_2$) and instead of $P_1$ the modified structure $S$ contains a path $P'_1=(s,f,c,x,g,h,x,a,c,b,a)$. Also, of course a hinge $h(c,b,d)$ is restored. 

We also want to point out that in the case of (standard) matchings, the structure $S$ has the following property: for any even length path $P$ of $S$ starting at $s$, each of its even-length subpaths $P'$ is also contained in $S$.  In the case of triangle-free $2$-matchings this property does not hold any more.
This happens in the following situations.
If a vulnerable triangle $t$ of type $1$ or $2$ becomes toppy by being involved in a blossom, then a vulnerable hinge is restored because of a path $P_v$ in $S$ that is non-standard in this sense. Suppose that $S$ at some point contains a path $P_1=(s, f,c,b,a)$ and $t=(a;b,c)$ is a triangle of type $2$. Then a hinge $h(c,b,d)$ is removed. At some later point $S$ contains a blossom $B=(s,f,c,b,a,x,y,s)$ and $t$ becomes toppy. $S$ will now contain a path $(s,y,x,a,c,b,d)$ going through a restored hinge. However, this path is not an extension of a path $(s,f,c,x^c_{bc}, x^b_{bc})$ which is also contained in $S$.

Moreover, $S$ is now allowed to slightly violate the following property mentioned before: If  a path $P_v$ uses at least one edge of a blossom $B$, then all edges of $P_v \cap B$ form one subpath $P'$ of $P$ such that
 one endpoint of $P'$ is the base of $B$. In our case, the satisfied property is slightly modified as follows. If the base of $B$ is $v_i$, then one endpoint of $P'$  is  $v_i$ or the other copy $v_{3-i}$ of $v$. If the base of $B$ is $x^a_{ab}$ and $t=(a;b,c)$ is a triangle of type $1$, then one endpoint of $P'$  is  $x^a_{ab}$ or $x^a_{ac}$. Two examples are shown and explained in Figure \ref{strukturaS}.

\section{Deferred proofs} \label{deferred}

{\bf Proof of Lemma \ref{triangle1}}

\dowod
 A path $P$ non-amenable on $t$ can essentially have one of the following four forms: \\
$(i) P=(s, \ldots, b',b, a,d, \ldots, e,a,c,c'), (ii) P=(s, \ldots, b',b, a,d, \ldots, c',c, a,e),$ \\ 
$(iii) P=(s, \ldots, d,a,b,b' \ldots, c',c, a,e), (iv) P=(s, \ldots, d,a,b,b', \ldots, e,a,c,c')$. The other forms of a path non-amenable on $P$ are symmetric. We will show that if $h(u,v,w)$ is a currently examined segment, then in $S \cup h^+(u,v,w)$ only the first form of a path can occur.

Suppose first that among all vertices of $t$, $B(a)$ is processed  first.
When $B(a)$ is processed it may form a singleton blossom but it may also belong to a larger blossom (but without any edges of $t$).
Let us assume that  $(d,a) \in S \cap M$.  Clearly, we can then extend $S$ to edges $(a,b), (b,b')$ and to edges $(a,c), (c,c')$ along paths amenable on $t$ and amenable on any other triangle $t'$. (A path $Q=(s, \ldots, d,a,b,b')$ cannot be non-amenable on a triangle $t' \neq t$ of type $1$ or $0$, because $t'$ would have to then be of the form $(b;e,a)$ or $(x,b,a)$.  In both cases $Q$ would  have to contain the edge $(b,c)$, because a path non-amenable on a triangle $t'$ of type $0$ or $1$ has to  contain, respectively, $4$ and $6$ edges of $M \setminus t$ incident to $t$.  Since $B(a)$ is processed first, it means that $Q$ does not contain $(b,c)$.) 
The only case when we cannot extend $S$ along the edges $(a,c), (c,b)$ is when $e=c'$ and thus $G$ contains a triangle $(e;a,c)$ of type $2$.
If $e$ and $c'$ denote different vertices, then $t$ forms a blossom in $S$ and we can also extend $S$ along the edge $(a,e)$ - so that $e \in \mathcal{E}(S)$. Hence $t$ never becomes  vulnerable.

Suppose now that $e=c'$, see Figure \ref{2triangles1} (b). In this case we remove the hinge $h(a,c,b)$ but $S$ is extended to the edge $(b,c)$ via the hinge $h(a,b,c)$. We notice that the hinge $h(c,a,c')$ will never become vulnerable because  if it is examined at some later point, $S$ will contain  a blossom that includes the edges $(a,c)$ and $(c,c')$ and hence $h(c,a,c')$ will be safe.
This shows that in this case $t$ never becomes vulnerable (and that a path of form (iii) cannot exist in $S \cup h^+(u,v,w)$.

Suppose next that $b$ is processed first. It can happen  because $S$  contains a path $P'$ ending with $(b',b)$ or ending with $(c,b)$.
Let us consider first the case when $P'$ ends with $(b',b)$. Let us notice that $P'$ does not contain $(b,c)$ because then $c$ would be processed first or the other copy of $b$ would be processed first. 

\begin{claim}
A path $P_1=P' \cup \{b,a,d\}$ is amenable. So is a path $P_2=P' \cup \{b,a,e\}$
\end{claim}
\dowod
Suppose to the contrary that $P_1$ is non-amenable on $t'$. The triangle $t'$ cannot be of type $2$ (because then $t'$ would have to be equal to $t$ - a contradiction.) If $t'$ were to be of type $1$, then it would have the form (1) $t'=(a;b',b)$ or (2) $t'=(b;a,e)$. In  case (1) $P_1$ cannot be non-amenable on it because it goes through $(b',b) \in M \cap t'$. In case (2)  $P_1$ cannot be non-amenable because it does not contain $(b,c)$, see Figure \ref{2triangles1} (a).  The triangle  $t'$ cannot be a triangle of type $0$ either because $P'$ does not contain $(b,c)$ and a path non-amenable on a triangle $t'$ of type $0$ has to contain $6$ edges of $M$ incident to $t'$. The proof for the path $P_2$ is symmetric.
\koniec

This already shows that in this case no hinge incident to $a$ can be removed and $S \cup h^+(u,v,w)$ cannot contain a path of form (ii) or of form (iv).

Let us consider next the case when $P'$ ends with $(b,c)$, see Figure \ref{2triangles1} (b). 
 In this case
a path $P' \cup \{c,a,d\}$ can be non-amenable on a triangle $t''$ only if $t''$ has the form $(c';a,c)$ and is of type $2$.
Here the reasoning is similar as when $a$ is processed first  and $e=c'$. We remove the hinge $h(a,c,d)$ from $G_2$ and extend $S$ along $(c,a), (a,c')$. The hinge $h(a,c,c')$ will never become vulnerable, because at the moment it starts to be $S$-augmenting, $S$ contains a blossom with the edges $(b,c), (c,a), (a,c')$, which means that $h(a,c,c')$ is safe and hence not vulnerable and thus $t$ cannot become vulnerable. 

If $t$ is toppy, then we can essentially proceed as in the case when $B(a)$ is processed first, because as we have shown above either no hinge incident to $a$ has been removed (and then we can extend a path $P_a$ from $s$ to $a$ that does not go through any edge of $t$ to the edges $(a,b),(b,b')$ and $(a,c), (c,c')$) or $t$ is guaranteed not to be vulnerable. Note that in the case when a vulnerable triangle of type $1$ becomes toppy, a restored hinge $h$ may be contained in a path $P_h$ of $S$ that goes through some blossom in a non-standard way - see, for example Figure \ref{strukturaS} (a). A similar situation happens for some vulnerable triangles of type $2$. We write more about it in Section \ref{corr}.
\koniec

\begin{figure}[h]
\centering{\includegraphics[scale=0.8]{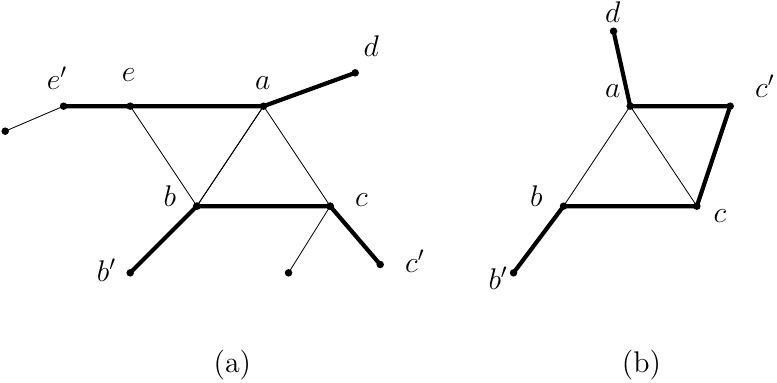}}
\caption{(a) Two triangles of type $1$ that share an edge connecting their top vertices. None of such triangles can become vulnerable. (b) A triangle $t_1=(a;b,c)$ of type $1$ sharing an edge with a triangle $t_2=(c';a,c)$ of type $2$.}
 \label{2triangles1}
\end{figure}

{\bf Proof of Lemma \ref{triangle0top}} \\
\dowod
If $t$ becomes vulnerable when we examine a hinge $h(u,v,w)$, then $S \cup h^+(u,v,w)$ must contain at this point a path $P_t$  of the form $(s, \ldots, f_1, e^t_1, f_2, \ldots,
f_3, e^t_2, f_4, \ldots, f_5, e^t_3, f_6)$, where  $e^t_1, e^t_2, e^t_3$ denote three edges of $t$, and $f_1, \ldots, f_6$ denote $6$ edges of $M'$ incident to $t$.

Suppose that an even length path $P$ of $S$ from $s$ to $a_1$ ends with $f_1=(d,a_1) \in M'$. When we process $a_1$, we check if the path $P$ from $s$ to $a_1$ can be extended with the edges $(a_1, u), (u,u')$, where $u$ is the vertex of $t$ and $(u,u')$ belongs to $M$. We argue that a path $P'=P \cup \{(a_1, u), (u,u')\}$ cannot be non-amenable on a triangle $t'$ of type $1$ or $0$. $P'$ cannot be non-amenable on a triangle $t'$ of type $0$ because either $t$ is the first triangle of type $0$ with a path in $S$ non-amenable on it, or the lemma currently being proved holds and then a path $P'$ can be non-amenable on a triangle $t'$ of type $0$ only if $S$ contains a blossom $B'$ containing all edges of $t$ and having base in $a_1$. Clearly, it is impossible for $S$ to contain such a blossom during the step when we examine a segment $\sigma=(a_1, u_i)$.

If $P'$ were to be non-amenable on $t'$ of type $1$, then by Lemma \ref{triangle1} the path $P'$ would have to go through a blossom incident to the top of $t'$. Hence, $S$ would have to contain a blossom $B'$ that contains two edges $f_i, f_j$ that are incident either to $a$ or to $c$. 
This would imply that $S$ contains a path that contains both $f_i$ and $f_j$ and none of the edges of $t$, which contradicts the possibility of the existence of $P_t$. This means that $P'$ can be non-amenable only on a triangle of type $2$ and therefore at most one of the paths $P_1=P \cup \{a,b,b'\}, P_2=P \cup \{a,b,b''\}, P_3= \cup \{a,c,c'\}, P_4=P \cup \{a,c,c''\}$ is non-amenable, where $(b,b'), (b,b''), (c,c'), (c,c''), (a,x)$ denote edges of $M$. 
Suppose it is $P \cup \{a,c,c''\}$ and  $t'=(x;a,c)$ is a triangle of type $2$.

Hence, when processing $a_1$, we extend $S$ so that it contains at least three of these paths or we form a blossom that contains  an  edge of $t$ incident to $a$.
In the second case we prove that $P_t$ cannot arise in $S$:

\begin{claim} \label{claim3}
If at some point $S$ contains a blossom $B$ with $i$  edges of $\{e^t_1, e^t_2\}$ and $1 \leq i \leq 2$, then $t$ is not vulnerable.
\end{claim}
\dowod
Suppose that $B$ contains an edge $e^t_1$ of $t$ and hence contains $f_1, e^t_1=(a_1,c), f_2$. $B$ then  does not contain the other edges $f_3, f_4$ of $M$ incident to $e^t_1$. 
Suppose that $f_3$ is incident to $a$ and $f_4$ to $c$. Then  the path  $P_2=(s, \ldots, f_2, e^t_1, f_3)$ is amenable and we can extend $S$.
The edges $f_1, e^t_1, f_2$ belong to a common blossom. Therefore, on any path of $S$ that contains all of them, they occur consecutively, in order $(f_1, e^t_1, f_2)$ 
or $(f_2, e^t_1, f_1)$. Since $S$ contains a path ending with $(f_2, e^t_1, f_3)$, thus a path of $S$ that contains all edges $f_1, f_2, f_3, e^t_1, e^t_2$ would have to have the form $(s, \ldots, f_1, e^t_1, f_2, \ldots, f_2, e^t_2)$. But we already know that $S$ contains a path $(s, \ldots, f_1, e^t_2)=P_1$ - a contradiction, because
$S$ may contain only one even-length path to any vertex and in particular to $x^a_{ab}$.
\koniec

If $S$ contains paths $P_1-P_3$, then in order for $S$ to contain $P_t$, $S$ has to contain a blossom $B'$ with the edges $(a,b)$ and $(a,c)$. (This is because $S$ has only one even-length path from $s$ to any vertex $u \in \mathcal{E}(S)$ and $S$  contains even length paths $P_1, P_2$ to $x^b_{ab}$ and $x^c_{ac}$, respectively, each of which contains exactly one of the edges $(x^a_{ab}, x^b_{ab}), (x^a_{ac}, x^c_{ac})$.) By Claim \ref{claim3}, $B'$ has to contain all edges of $t$. Let us also notice that such a blossom cannot have the form $B''=(a,b, \ldots, c,b, \ldots, c,a)$, i.e., after traversing the edge $(a,b)$ it cannot jump to the endpoint $c$ of the edge $(c,b)$ and then go through some edges and reach the endpoint $c$ of the edge $(c,a)$. This is because $S$
contains the paths $P_1$ and $P_2$ and hence on any path in $S$ the edges $(b,b'), (b,b'')$ can be traversed in the order $(b,b', \ldots, b'', b)$ or in the order opposite to it (again because $S$ has only one even-length path from $s$ to any vertex $u \in \mathcal{E}(S)$), i.e., before forming $B'$, a blossom $B(b)$ is formed that includes the edges $(b,b'), (b,b'')$. This blossom $B(b)$ has its base in $x^b_{ab}$.

Let us finally notice that if $S$ contains a blossom containing all edges of $t$ and $t$ shares an edge with a triangle $t'=(c';a,c)$ of type $2$, then $t$ is not vulnerable,
because the hinges $h(c,a,c')$ and $h(b,a,c')$ are safe, compare Figure \ref{zawiasywewfig}, in which the hinges $h(c,b,a), h(i,b,a)$ of $t=(b;c,i)$ are impassable but  not vulnerable.
\koniec

{\bf Proof of Lemma \ref{safe}.} \\

\dowod 

Let $t=(a;b,c)$ be a triangle of type $2$ such that $M' \cap t=\{(a_1, c_1), (a_2, b_1)\}$. Assume that the algorithm at some point discovers that it could form a blossom that consists only of the edge $(a_2, b_1)$. This can happen in two cases: (i) either $b_2$ and $a_1$ belong to the same blossom or (ii) the algorithm is examining a segment $\sigma=(a_1, v_i)$ such that $v_i \in B(b_1)$ and of course, $a_1 \notin B(b_1)$. Let us first deal with case (i).
Suppose also that $(b_2, d_2)$ is an edge of $M' \setminus t$.

Let us first deal with case (i). Suppose also that $(b_2, d_2)$ is an edge of $M' \setminus t$.
  It suffices to prove that no hinge incident to $b_2$ or  to $a_1$ is impassable or that even if it is impassable, its addition would not increase the alternating s-reachability of $S$. We show that no hinge incident to $b_2$ is impassable. Suppose to the contrary that some hinge $h(y,b,d)$ is impassable on a triangle $t''$.  By Observation \ref{zawiasy} we know that $t''$ cannot be a triangle of type $2$. If $t''$ were to be of type $1$, then it would have to have the form $t''=(y;a,b)$. However, since $b_2$ and $a_1$ belong to the same blossom, the hinge $h(y,b,d)$ is not $S$-augmenting and hence, cannot be impassable.
Suppose now that some hinge incident to $a_1$ is impassable, i.e., assume that some hinge $h(y,a,c)$ is impassable on a triangle $t''$. Then $t''$ cannot be of type $1$ because as we have observed above $h(y,a,c)$ cannot be $S$-augmenting because $b_2$ and $a_1$ belong to the same blossom.  If $t''=(b;a,d)$ is of type $2$, then $h(y,a,c)$ may be impassable on $t''$. In this case, however, we prove that no hinge incident to $d_2$ can be impassable on any triangle $t_3$. This follows from the following. By Observation \ref{zawiasy} $t_3$ cannot be of type $2$ and $t_3$ cannot be of type $1$ because $S$ does not contain a path $(s, \ldots, a,b)$ that does not go through $(d,z) \in M \setminus t''$.

The proof of case (ii) is  contained in the second part of the proof of case (i). Case (ii) may arise in situations shown in Figure \ref{safespecial}. If we first consider all edges except for $(d,a), (d,k), (k,b)$, then the hinge $h(c,b,a)$ is impassable and we do not add it to $S$. However, we form a blossom containing all edges of $t=(b,c,i)$ and edges $(c,a), (c,f), (a,f), (i,j), (i,g), (j,g)$. Further, we add edges $(d,k), (k,b)$ as well as the hinge $h(a,d,b)$ to $S$. At the point when we consider the hinge $h(d,a,c)$, we notice that it is impassable but safe, because its addition to $S$, would mean that we create a new blossom, whose only new edge of $M'$ would be $(d,b)$.

\begin{figure}
\centering{\includegraphics[scale=0.8]{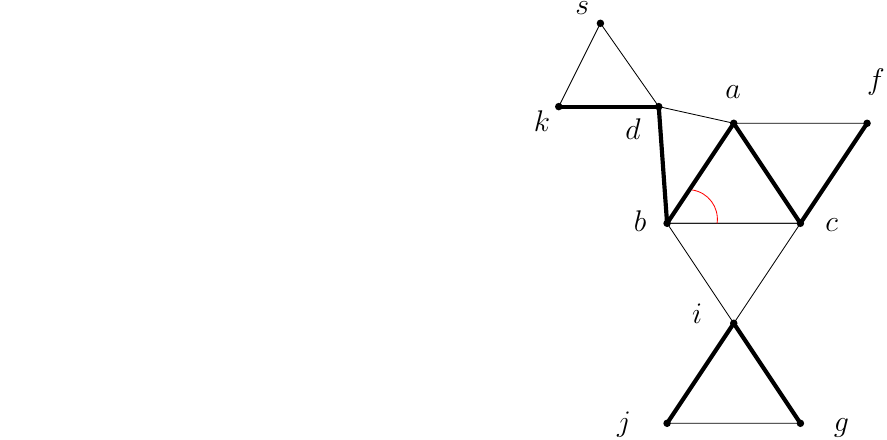}}
\caption{}
 \label{safespecial}
\end{figure}
\koniec

\section{Acknowledgements} I would like to thank Kristóf B{\'e}rczi for many interesting discussions. I also thank Mateusz Wasylkiewicz for comments on Section \ref{decomp} and anonymous reviewers for their helpful feedback.

\bibliographystyle{abbrv}
\bibliography{bib}

\begin{thebibliography}{10}

\bibitem{AdamaszekMP18}
A.~Adamaszek, M.~Mnich, and K.~Paluch.
\newblock New approximation algorithms for (1, 2)-tsp.
\newblock In I.~Chatzigiannakis, C.~Kaklamanis, D.~Marx, and D.~Sannella,
  editors, {\em 45th International Colloquium on Automata, Languages, and
  Programming, {ICALP} 2018, July 9-13, 2018, Prague, Czech Republic}, volume
  107 of {\em LIPIcs}, pages 9:1--9:14. Schloss Dagstuhl - Leibniz-Zentrum
  f{\"{u}}r Informatik, 2018.

\bibitem{ArtamonovBabenko2018}
S.~Artamonov and M.~Babenko.
\newblock A fast scaling algorithm for the weighted triangle-free 2-matching
  problem.
\newblock {\em European Journal of Combinatorics}, 68:3 -- 23, 2018.

\bibitem{BabenkoEtAl2010}
M.~Babenko, A.~Gusakov, and I.~Razenshteyn.
\newblock Triangle-free 2-matchings revisited.
\newblock In {\em Computing and Combinatorics}, pages 120--129, 2010.

\bibitem{BercziKobayashi2012}
K.~B{\'e}rczi and Y.~Kobayashi.
\newblock An algorithm for {$(n-3)$}-connectivity augmentation problem: Jump
  system approach.
\newblock {\em Journal of Combinatorial Theory, Series B}, 102(3):565--587,
  2012.

\bibitem{BercziVegh2010}
K.~B{\'e}rczi and L.~V{\'e}gh.
\newblock Restricted b-matchings in degree-bounded graphs.
\newblock In {\em Integer Programming and Combinatorial Optimization}, pages
  43--56, 2010.

\bibitem{BermanKarpinski2006}
P.~Berman and M.~Karpinski.
\newblock {$8/7$}-approximation algorithm for {$(1,2)$}-{TSP}.
\newblock In {\em Proc. SODA 2006}, pages 641--648, 2006.

\bibitem{BlaserRam2005}
M.~Bl{\"a}ser and L.~S. Ram.
\newblock An improved approximation algorithm for {TSP} with distances one and
  two.
\newblock In {\em Proc. FCT 2005}, volume 3623 of {\em LNCS}, pages 504--515.
  2005.

\bibitem{Calvostoc}
M.~Bosch{-}Calvo, M.~Garg, F.~Grandoni, F.~Hommelsheim, A.~J. Ameli, and
  A.~Lindermayr.
\newblock A 5/4 approximation for two-edge-connectivity.
\newblock {\em CoRR}, abs/2408.07019, 2024, accepted to STOC 2025.

\bibitem{Bosch}
M.~Bosch{-}Calvo, F.~Grandoni, and A.~J. Ameli.
\newblock A {PTAS} for triangle-free 2-matching.
\newblock {\em CoRR}, abs/2311.11869, 2023.

\bibitem{CornuejolsPulleyblank1980}
G.~Cornu\'ejols and W.~Pulleyblank.
\newblock A matching problem with side conditions.
\newblock {\em Discrete Mathematics}, 29(2):135--159, 1980.

\bibitem{CornuejolsPulleyblank1980a}
G.~Cornu{\'e}jols and W.~Pulleyblank.
\newblock Perfect triangle-free 2-matchings.
\newblock In {\em Combinatorial Optimization II}, Mathematical Programming
  Studies, pages 1--7. Springer Berlin Heidelberg, 1980.

\bibitem{Duan}
R.~Duan and S.~Pettie.
\newblock Linear-time approximation for maximum weight matching.
\newblock {\em J. {ACM}}, 61(1):1:1--1:23, 2014.

\bibitem{FisherEtAl1979}
M.~L. Fisher, G.~L. Nemhauser, and L.~A. Wolsey.
\newblock An analysis of approximations for finding a maximum weight
  hamiltonian circuit.
\newblock {\em Operations Research}, 27(4):799--809, 1979.

\bibitem{Gabow17}
H.~N. Gabow.
\newblock The weighted matching approach to maximum cardinality matching.
\newblock {\em Fundam. Informaticae}, 154(1-4):109--130, 2017.

\bibitem{GabowTarjan}
H.~N. Gabow and R.~E. Tarjan.
\newblock A linear-time algorithm for a special case of disjoint set union.
\newblock {\em J. Comput. Syst. Sci.}, 30(2):209--221, 1985.

\bibitem{Geelen1999}
J.~Geelen.
\newblock The {$C_6$}-free {$2$}-factor problem in bipartite graphs is
  {NP}-complete.
\newblock Unpublished, 1999.

\bibitem{Hartvigsen1984}
D.~Hartvigsen.
\newblock {\em Extensions of Matching Theory}.
\newblock PhD thesis, Carnegie-Mellon University, 1984.

\bibitem{Hartvigsen2006}
D.~Hartvigsen.
\newblock Finding maximum square-free {$2$}-matchings in bipartite graphs.
\newblock {\em Journal of Combinatorial Theory, Series B}, 96(5):693--705,
  2006.

\bibitem{Hartvigsen2024}
D.~Hartvigsen.
\newblock Finding triangle-free 2-factors in general graphs.
\newblock {\em Journal of Graph Theory}, 106(3):1--82, 2024.

\bibitem{HartvigsenLi2012}
D.~Hartvigsen and Y.~Li.
\newblock Polyhedron of triangle-free simple 2-matchings in subcubic graphs.
\newblock {\em Mathematical Programming}, 138:43--82, 2013.

\bibitem{Edmonds}
J.Edmonds.
\newblock Paths, trees and flowers.
\newblock {\em Canad. J. Math.}, (17):449--467, 1965.

\bibitem{Kiraly1999}
Z.~Kir{\'a}ly.
\newblock {$C_4$}-free {$2$}-factors in bipartite graphs.
\newblock Technical report, Egerv{\'a}ry Research Group, 1999.

\bibitem{Kiraly2009}
Z.~Kir{\'a}ly.
\newblock Restricted {$t$}-matchings in bipartite graphs.
\newblock Technical report, Egerv{\'a}ry Research Group, 2009.

\bibitem{Kobayashi2010}
Y.~Kobayashi.
\newblock A simple algorithm for finding a maximum triangle-free {$2$}-matching
  in subcubic graphs.
\newblock {\em Discrete Optimization}, 7:197--202, 2010.

\bibitem{Kobayashi2020}
Y.~Kobayashi.
\newblock Weighted triangle-free {$2$}-matching problem with edge-disjoint
  forbidden triangles.
\newblock In {\em Integer Programming and Combinatorial Optimization}, pages
  280--293, 2020.

\bibitem{Kobayashi2023}
Y.~Kobayashi and T.~Noguchi.
\newblock An approximation algorithm for two-edge-connected subgraph problem
  via triangle-free two-edge-cover.
\newblock In {\em 34th International Symposium on Algorithms and Computation,
  {ISAAC} 2023, December 3-6, 2023, Kyoto, Japan}, volume 283 of {\em LIPIcs},
  pages 49:1--49:10. Schloss Dagstuhl - Leibniz-Zentrum f{\"{u}}r Informatik,
  2023.

\bibitem{NoguchiSosa}
Y.~Kobayashi and T.~Noguchi.
\newblock Validating a {PTAS} for triangle-free 2-matching via a simple
  decomposition theorem.
\newblock In {\em 2025 Symposium on Simplicity in Algorithms, {SOSA} 2025, New
  Orleans, LA, USA, January 13-15, 2025}, pages 281--289. {SIAM}, 2025.

\bibitem{LovaszPlummer2009}
L.~Lov{\'a}sz and M.~Plummer.
\newblock {\em Matching theory}.
\newblock AMS Chelsea Publishing, corrected reprint of the 1986 original
  edition, 2009.

\bibitem{Makai2007}
M.~Makai.
\newblock On maximum cost {$K_{t,t}$}-free {$t$}-matchings of bipartite graphs.
\newblock {\em {SIAM} Journal on Discrete Mathematics}, 21:349--360, 2007.

\bibitem{Nam1994}
Y.~Nam.
\newblock {\em Matching Theory: Subgraphs with Degree Constraints and other
  Properties}.
\newblock PhD thesis, University of British Columbia, 1994.

\bibitem{PaluchEtAl2012}
K.~Paluch, K.~Elbassioni, and A.~van Zuylen.
\newblock Simpler approximation of the maximum asymmetric traveling salesman
  problem.
\newblock In {\em 29th International Symposium on Theoretical Aspects of
  Computer Science}, pages 501--506, 2012.

\bibitem{PaluchWasylESA}
K.~Paluch and M.~Wasylkiewicz.
\newblock Restricted t-matchings via half-edges.
\newblock In P.~Mutzel, R.~Pagh, and G.~Herman, editors, {\em 29th Annual
  European Symposium on Algorithms, {ESA} 2021, September 6-8, 2021, Lisbon,
  Portugal (Virtual Conference)}, volume 204 of {\em LIPIcs}, pages
  73:1--73:17. Schloss Dagstuhl - Leibniz-Zentrum f{\"{u}}r Informatik, 2021.

\bibitem{PaluchWasylkiewicz2020}
K.~Paluch and M.~Wasylkiewicz.
\newblock A simple combinatorial algorithm for restricted 2-matchings in
  subcubic graphs - via half-edges.
\newblock {\em Information Processing Letters}, 171, 2021.

\bibitem{Pap2007}
G.~Pap.
\newblock Combinatorial algorithms for matchings, even factors and square-free
  2-factors.
\newblock {\em Mathematical Programming}, 110:57--69, 2007.

\bibitem{Pap2009}
G.~Pap.
\newblock Weighted restricted 2-matching.
\newblock {\em Mathematical Programming}, 119:305--329, 2009.

\bibitem{Russell2001}
M.~Russell.
\newblock Restricted two-factors.
\newblock Master's thesis, University of Waterloo, 2001.

\bibitem{Schrijver2003}
A.~Schrijver.
\newblock {\em Combinatorial optimization: polyhedra and efficiency},
  volume~24.
\newblock Springer Science \& Business Media, 2003.

\bibitem{Takazawa2009}
K.~Takazawa.
\newblock A weighted {$K_{t,t}$}-free {$t$}-factor algorithm for bipartite
  graphs.
\newblock {\em Mathematics of Operations Research}, 34(2):351--362, 2009.

\bibitem{Takazawa2017}
K.~Takazawa.
\newblock Excluded t-factors in bipartite graphs: A unified framework for
  nonbipartite matchings and restricted 2-matchings.
\newblock In {\em Integer Programming and Combinatorial Optimization}, pages
  430--441, 2017.

\bibitem{Takazawa2017a}
K.~Takazawa.
\newblock Finding a maximum 2-matching excluding prescribed cycles in bipartite
  graphs.
\newblock {\em Discrete Optimization}, 26:26--40, 2017.

\bibitem{Vornberger1980}
O.~Vornberger.
\newblock Easy and hard cycle covers.
\newblock Technical report, Universit\"at Paderborn, 1980.

\end{thebibliography}


\end{document}